\documentclass[11pt]{article}
\usepackage[bookmarks,bookmarksopen,bookmarksnumbered,
colorlinks,
pdfpagemode = {UseOutlines},
pdfstartview = {FitH},
urlcolor = {blue}]{hyperref}
\usepackage[backend=bibtex,style=numeric,maxnames=5]{biblatex}
\usepackage[lining,semibold]{libertine} 
\usepackage[T1]{fontenc} 
\usepackage{textcomp} 
\usepackage[varqu,varl]{inconsolata}
\usepackage{amsmath,amssymb,amsthm}
\usepackage{dutchcal}
\usepackage[libertine,cmintegrals,bigdelims,vvarbb]{newtxmath}
\usepackage[scr=rsfso]{mathalfa}
\usepackage{bm}
\usepackage{float}
\usepackage{color}
\usepackage{graphicx}
\usepackage{amsmath,amssymb,verbatim,bm}
\usepackage{ctable} 
\usepackage{slashed}
\usepackage{longtable}
\usepackage{url}
\usepackage{bm}
\usepackage[english]{babel}
\usepackage{microtype}		
\usepackage{float}
\usepackage{color}
\usepackage{graphicx}
\usepackage{amsmath,amssymb,verbatim}
\usepackage[layout=letterpaper,top=25.0mm, bottom=25.0mm, left=25.0mm, right=25.0mm]{geometry}
\usepackage[normalem]{ulem}
\definecolor{violet}{RGB}{124,35,217}

\numberwithin{figure}{section}
\numberwithin{equation}{section}
\newcommand{\hspc[1]}{\hspace{#1pt}}

\newcommand{\ER}{Erd\H{o}s-R\'enyi }
\DeclareMathOperator{\dv}{div}
\DeclareMathOperator{\RP}{RP}
\DeclareMathOperator{\RD}{RD}
\DeclareMathOperator{\proba}{Prob}

\DeclareRobustCommand{\o}[1]{{\cal o} \left(#1\right)}
\DeclareRobustCommand{\O}[1]{{\cal O} \left(#1\right)}
\DeclareRobustCommand{\T}[1]{{\Theta} \left(#1\right)}
\DeclareRobustCommand{\om}[1]{{\omega} \left(#1\right)}
\DeclareRobustCommand{\dbar}{{\overline{d_n}}}
\DeclareRobustCommand{\dbarp}{{\overline{d_{n+1}}}}
\DeclareRobustCommand{\E}[1]{\mathbb{E}\left[#1\right]}
\DeclareRobustCommand{\prob}[1]{\proba\left(#1\right)}
\newtheorem{Theorem}{Theorem}
\newtheorem{proposition}{Proposition}

\newtheorem{lemma}{Lemma}
\newtheorem{definition}{Definition}

\newtheorem{remark}{Remark}

\newtheorem{corollary}{Corollary}

\newcommand{\N}{\mathbb{N}}

\newcommand{\bA}{\bm{A}}
\newcommand{\bB}{\bm{B}}
\newcommand{\bD}{\bm{D}}
\newcommand{\bL}{\bm{L}}

\newcommand{\cG}{\mathcal{G}}

\newcommand{\hR}{\widehat{R}}
\newcommand{\DhR}{\Delta \hR}
\DeclareRobustCommand{\er}[1]{{\hR_{#1}}}
\newcommand{\hbR}{{\widehat{\bm{R}}}}

\newcommand{\eqdef}{\stackrel{\text{\tiny def}}{=}}

\title{Detecting Topological Changes in Dynamic Community Networks
\footnote{This work was supported by NSF DMS 1407340.}}
\author{Peter Wills and Fran\c{c}ois G. Meyer\\
  \url{peter.wills@colorado.edu}, \url{fmeyer@colorado.edu}\\
  ~\\
  Applied Mathematics \& Electrical Engineering\\
  University of Colorado at Boulder, Boulder CO 80305}
\begin{document}
\date{}
\maketitle
\begin{abstract}
  The study of time-varying (dynamic) networks (graphs) is of fundamental importance for computer
  network analytics. Several methods have been proposed to detect the effect of significant structural 
  changes in a time series of graphs. 

  The main contribution of this work is a detailed analysis of a dynamic community graph model. This
  model is formed by adding new vertices, and randomly attaching them to the existing nodes. It is a
  dynamic extension of the well-known stochastic blockmodel. The goal of the work is to detect the
  time at which the graph dynamics switches from a normal evolution -- where balanced communities
  grow at the same rate -- to an abnormal behavior -- where communities start merging.

  In order to circumvent the problem of decomposing each graph into communities, we use a metric to
  quantify changes in the graph topology as a function of time.  The detection of anomalies becomes
  one of testing the hypothesis that the graph is undergoing a significant structural change.

  In addition the the theoretical analysis of the test statistic, we perform Monte Carlo simulations
  of our dynamic graph model to demonstrate that our test can detect changes in graph topology.
\end{abstract}

\section{Introduction}
The study of time-varying (dynamic) networks (or graphs) is of fundamental importance for
computer network analytics and the detection of anomalies associated with cyber crime
\cite{hoque14,ide04,king14}. Dynamic graphs also provide models for social networks
\cite{adler07,gilbert11}, and are used to decode the functional connectivity in neuroscience
\cite{gollo14,hutt14,stam14} and biology \cite{bass13}. The significance of this research
topic has triggered much recent work \cite{akoglu14,lafond14,ranshous15}. Several methods have
been proposed to detect the effect of significant structural changes (e.g., changes in topology,
connectivity, or relative size of the communities in a community graph) in a time series of
graphs. We focus on networks that change over time, allowing both edges and nodes to be added or
removed. We refer to these as \emph{dynamic networks}.

A fundamental goal of the study of dynamic graphs is the identification of universal patterns that
uniquely couple the dynamical processes that drive the evolution of the connectivity with the
specific topology of the network; in essence the discovery of universal spatio-temporal patterns
\cite{kovanen13,karsai14}. In this context, the goal of the present work is to detect anomalous changes
in the evolution of dynamic graphs. We propose a novel statistical method, which captures the
coherence of the dynamics under baseline (normal) evolution of the graph, and can detect switching
and regime transitions triggered by anomalies.  Specifically, we study a mathematical model of
normal and abnormal growth of a community network. Dynamic community networks have recently been the
topic of several studies \cite{toivonen06,baingana16,chi07,fu09,kolar10,lin08,mitra12}. The simplest
incarnation of such models, a dynamic stochastic blockmodel
\cite{tang14,xing10,xu14,xu15,yang11,wilson16,pensky16,ho15}, is the subject of our study. These
graph models have a wide range of applications, ranging from social networks
\cite{ishiguro10,yu16,lin09,tantipathananandh07,palla07,greene10,heaukulani13,laurent15} to computer
networks \cite{pincombe07,sun10} and even biology and neuroscience \cite{lee17}.

In order to circumvent the problem of decomposing each graph into simpler structures (e.g.,
communities), we use a metric to quantify changes in the graph topology as a function of time.
The detection of anomalies becomes one of testing the hypothesis that the graph is undergoing a
significant structural change. Several notions of similarity have been proposed to quantify the
structural similitude without resorting to the computation of a true distance (e.g.,
\cite{baur05,koutra16} and references therein). Unlike a true metric, a similarity is
typically not injective (two graphs can be perfectly similar without being the same), and rarely
satisfies the triangle inequality. This approach relies on the construction of a feature vector
that extracts a signature of the graph characteristics; the respective feature vectors of the two
graphs are then compared using a distance, or a kernel. In the extensive review of Koutra et
al. \cite{koutra16}, the authors studied several graph similarities and distances. They concluded
that existing similarities and distances either fail to conform to a small number of well-founded
axioms, or suffer from a prohibitive computational cost.  In response to these shortcomings,
Koutra et al. proposed a novel notion of similarity \cite{koutra16}.   

Inspired by the work of \cite{koutra16}, we proposed in \cite{monnig16} a true metric that
address some of the limitations of the DeltaCon similarity introduced in \cite{koutra16}.  We
emphasize that it is highly preferable to have a proper metric, rather than an informal distance,
when comparing graphs; this allows one to employ proof techniques not available in the absence of
the triangle inequality.  Our distance, coined the \emph{resistance-perturbation distance}, can
quantify structural changes occurring on a graph at different scales: from the local scale formed
by the neighbors of each vertex, to the largest scale that quantifies the connections between
clusters, or communities. Furthermore, we proposed fast (linear in the number of edges) randomized
algorithms that can quickly compute an approximation to the graph metric, for which error bounds
are proven (in contrast to the {\tt DeltaCon} algorithm given in \cite{koutra16}, which has a
linear time approximate algorithm but for which no error bounds are given). 

The main contribution of this work is a detailed analysis of a dynamic community graph model, which
we call the dynamic stochastic blockmodel. This model is formed by adding new vertices, and randomly
attaching them to the existing nodes. The goal of the work is to detect the time at which the graph
dynamics switches from a normal evolution -- where two balanced communities grow at the same rate --
to an abnormal behavior -- where the two communities are merging. Because the evolution of the graph
is stochastic, one expects random fluctuations of the graph geometry. The challenge is to detect an
anomalous event under normal random variation.  We propose an hypothesis test to detect the abnormal
growth of the balanced stochastic blockmodel.  In addition to the theoretical analysis of the test
statistic, we conduct several experiments on synthetic networks, and we demonstrate that our test
can detect changes in graph topology.

The remainder of this paper is organized as follows.  In the next section we introduce the main
mathematical concepts and corresponding nomenclature. In section \ref{rp-section} we recall the
definition of the resistance perturbation distance. We provide a straightforward extension of the
metric to graphs of different sizes and disconnected graphs. In section \ref{the-models} we formally
define the problem, we introduce the dynamic \emph{balanced two-community stochastic blockmodel}. We
describe the main contributions and the line of attack in Section \ref{the-main-results}. In Section
\ref{experiments} we present the results of experiments conducted on synthetic dynamic networks,
followed by a short discussion in Section \ref{discussion}.
\section{Preliminaries and Notation}
We denote by $G = (V,E)$ an undirected, unweighted graph, where $V$ is the vertex set of size $n$,
and $E$ is the edge set of size $m$. We will often use $u,$ $v,$ or $w$ to denote vertices in
$V$. For an edge $e\in E$, we denote by $\text{endpoints} \,(e)$ the subset of $V$ formed by the two
endpoint of $e$.

We use the standard asymptotic notation; see Appendix \ref{appx:asymptotic} for details.  Given
a family of probability spaces $\Omega = \left\{\Omega_n,\proba_n\right\}$, and a sequence of events
$E=\left\{E_n\right\}$, we write that $\Omega$ has the property \emph{with high probability}
(``w.h.p.''), if $\lim_{n\rightarrow \infty} \prob{E_n} =1 $.\\

\noindent When there is no ambiguity, we use the following abbreviated summation notation,
\begin{equation*}
  \label{eq:36}
  \sum_{u\leq n} \text{ is short for } \sum_{u=1}^n \text{ and } \sum_{u<v\leq n} \text{ is short for } \sum_{u=1}^n \sum_{v=u+1}^n.
\end{equation*}

Table \ref{tab:notations} in Appendix \ref{the-notations} provides a list of the main notations used in the paper.
\subsection{Effective Resistance}
We briefly review the notion of effective resistance \cite{klein1993,doyle1984,ghosh2008,ellens2011}
on a connected graph. The reader familiar with the concept can jump to the next section. There are
many different ways to present the concept of effective resistance. We use the electrical analogy,
which is very standard (e.g., \cite{doyle1984}). Given an unweighted graph $G=(V,E)$, we transform
$G$ into a resistor network by replacing each edge $e$ by a resistor with unit resistance.
\begin{definition}[Effective resistance \cite{klein1993}]
  The effective resistance $\er{uv}$ between two vertices $u$ and $v$ in $V$ is defined as the
  voltage applied between $u$ and $v$ that is required to maintain a unit current through the
  terminals formed by $u$ and $v$.\\

  \noindent We denote by $\hbR$ the $n\times n$ matrix with entries $\er{uv}, \, u,v =1,\ldots,n$.
  \label{eff-res}
\end{definition}
The relevance of the effective resistance in graph theory stems from the fact that it provides a
distance on a graph \cite{klein1993} that quantifies the connectivity between any two vertices, not
simply the length of the shortest path. Changes in effective resistance reveal structural changes
occurring on a graph at different scales: from the local scale formed by the neighbors of each
vertex, to the largest scale that quantifies the connections between clusters, or communities.

\section{Resistance Metrics\label{rp-section}}
\subsection{The Resistance Perturbation Metric}
The effective resistance can be used to track structural changes in a graph, and we use it to define
a distance between two graphs on the same vertex set \cite{monnig16} (see also \cite{sricharan14}
for a similar notion of distance). Formally, we define the Resistance Perturbation Distance as
follows.
\begin{definition}[Resistance Perturbation Distance]
  Let $G^{(1)} = (V,E^{(1)})$ and $G^{(2)} = (V,E^{(2)})$ be two connected, unweighted, undirected
  graphs on the same vertex set, with respective effective resistance matrices, $\hbR^{(1)}$ and
  $\hbR^{(2)}$ respectively.  The {\em RP-p distance} between $G^{(1)}$ and $G^{(2)}$ is defined as
  the element-wise p-norm of the difference between their effective resistance matrices.  For
  $1 \leq p < \infty$,
  \begin{equation}
    \RP_p(G^{(1)},G^{(2)}) = \left\lVert \hbR^{(1)} - \hbR^{(2)} \right\rVert_{p} =  \left[
      \sum_{i,j \in V} \left\lvert \er{ij}^{(1)} - \er{ij}^{(2)} \right\rvert^p
    \right]^{1/p}\mspace{-24mu}. 
    \label{rp_dist_def_eqn}
  \end{equation}
\end{definition}
In this paper, we will restrict our attention to the $\RP_1$ distance (we will omit the
subscript $p=1$), because it is directly analogous to the Kirchhoff index.
\subsection{Extending the Metric to Disconnected Graphs}
The resistance metric is not properly defined when the vertices are not within the same connected
component. To remedy this, we use a standard approach. Letting $\er{uv}$ denote the effective
resistance between two vertices $u$ and $v$ in a graph, then the conductivity
$C_{u\, v} = \er{uv}^{-1}$ can be defined to be zero for vertices in disconnected
components. Considering the conductivity as a similarity measure on vertices, a distance is given by
the quantity $(1+C_{u\, v})^{-1}$. Note that $(1+C_{u\, v})^{-1} = \hR_{u\, v}/(\hR_{u\, v}+1)$,
and so we can define this new quantity relative to the effective resistance without any reference to
the conductance. We refer to the resulting quantity as the \emph{renormalized effective resistance}.

\begin{definition} [Renormalized Effective Resistance]
  Let $G=(V,E)$ be a graph (possibly disconnected).  

  We define the renormalized effective resistance between any two vertices $u$ and $v$ to be
  \begin{equation}
    R_{u\, v} = 
    \begin{cases}
      \hR_{u\, v}/(\hR_{u\, v}+\beta) & \text{if $u$ and $v$ are connected},\\
      1 & \text{otherwise,}
    \end{cases}
    \label{eq:renorm-res-defn} 
  \end{equation} 
  where $\hR_{u\, v}$ is the effective resistance between $u$ and $v$, and $\beta >0$ is an arbitrary
  constant.

\end{definition}
We now proceed to extend the notion of resistance perturbation distance. 
\begin{definition} [Renormalized Resistance Distance]
  Let $G^{(1)}=(V^{(1)},E^{(1)})$ and $G^{(2)}=(V^{(2)},E^{(2)})$ be two graphs (with possibly different vertex sets).  We
  consider $V = V^{(1)}\cup V^{(2)}$, and relabel the union of vertices using $[n]$, where $n=|V|$. Let
  $R^{(1)}$ and $R^{(2)}$ denote the renormalized effective resistances in
  $\widetilde{G^{(1)}} = (V,E^{(1)})$ and $\widetilde{G^{(2)}} = (V,E^{(2)})$ respectively.

  We define the renormalized resistance distance to be 
  \begin{equation}
    \RD_\beta(G^{(1)},G^{(2)}) = \sum_{u<v\leq n} \left|   R^{(1)}_{u\, v} - R^{(2)}_{u\, v} \right|.
    \label{eq:metric-defn}  
  \end{equation}
  where the parameter $\beta$ (see \eqref{eq:renorm-res-defn}) is implicitly defined. In the rest of
  the paper we work with $\beta =1$, and dispense of the subscript $\beta$ in
  \eqref{eq:metric-defn}. In other words,  
  \begin{equation}
    \RD \eqdef \RD_1.
    \label{eq:RD}
  \end{equation}
  \label{defn:metric} 
\end{definition}
\begin{remark}
  An additional parameter $\beta$ has been added to the definition. Changing $\beta$ is
  equivalent to scaling the effective resistance before applying the function
  $x \rightarrow x / (1+x)$. Note that when $\hR\ll \beta$, then
  $R \approx \hR / \beta$, i.e. the renormalized resistance is approximately a rescaling of the
  effective resistance. Note that in this metric, two graphs are equal if they differ only in addition
  or removal of isolated vertices.
\end{remark}
\noindent The following lemma confirms that the distance defined by \eqref{eq:metric-defn} remains a
metric when we compare graphs with the same vertex set.
\begin{lemma}
  Let $V$ be a vertex set. $\RD$ defined by \eqref{eq:metric-defn}  us a metric on the space of
  unweighted undirected graphs defined on the same vertex set $V$.
\end{lemma}

\begin{remark}
  The metric given in Definition \ref{defn:metric} can be used to compare graphs of two different
  sizes, by adding isolated vertices to both graphs until they have the same vertex set (this is why
  we must form the union $V = V^{(1)} \cup V^{(2)}$ and compare the graphs over this vertex set). This method
  will give reasonable results when the overlap between $V^{(1)}$ and $V^{(2)}$ is large. In particular, if we
  are comparing graphs of size $n$ and $n+1$, then we only need add one isolated vertex to the former
  so that we can compare it to the latter. This situation is illustrated in
  Figure~\ref{fig:new-vertex}.
\end{remark}
\begin{figure}[H]
  \centering{
    \includegraphics[width=0.4\textwidth]{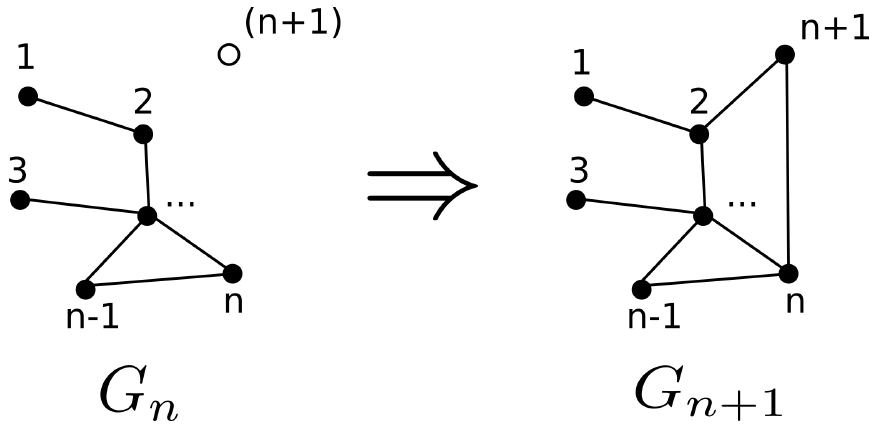}
    
    \caption{In order to compare $G_n$ and $G_{n+1}$, we include node $n+1$ into $G_n$ (see left),
      and evaluate the renormalized effective resistance on the augmented graph, with vertex set
      $\left\{1,\ldots,n\right\}\cup \left\{n+1\right\}$.}
    \label{fig:new-vertex}}
\end{figure}
When the graphs $G^{(1)}$ and $G^{(2)}$ have different sizes, the distance $\RD$ still satisfies the
triangle inequality, and is symmetric. However, $\RD$ is no longer injective: it is a
pseudo-metric. Indeed, as explained in the following lemmas, if $\RD(G^{(1)},G^{(2)}) = 0$, then the
connected components of $G^{(1)}$ and $G^{(2)}$ are the same, but the respective vertex sets may
differ by an arbitrary number of isolated vertices.
\begin{lemma}
  Let $G=(E,V)$ be an unweighted undirected graph, and let $V^{(i)}$ be a set of isolated vertices, to wit
  $V^{(i)} \cap V = \emptyset$ and $\forall e \in E, \text{endpoints} \,(e) \notin V^{(i)}$. Define $G^\prime
  =(V \cup V^{(i)}, E)$, then we have $\RD(G,G^\prime) = 0$.
\end{lemma}
\noindent The following lemma shows that the converse is also true.
\begin{lemma}
  Let $G^{(1)} = (V,E^{(1)})$ and $G^{(2)} = (V,E^{(2)})$ be two unweighted, undirected
  graphs, where $|V^{(1)}| > |V^{(2)}$. 

  \noindent If $\RD(G^{(1)}, G^{(2)}) = 0$, then $E^{(1)} = E^{(2)}$. Furthermore, there exists a set
  $V^{(i)}$ of isolated vertices, such that $V^{(1)} = V^{(2)} \cup V^{(i)}$.
\end{lemma}
\noindent In summary, in this work the distance $\RD$ will always be a metric since we will only consider graphs
that are connected with high probability.
\section{Graph Models
  \label{the-models}}
In our analysis, we will discuss two common random graph models, the classic model of Erd\H{o}s
and R\'enyi \cite{bollobas13} and the stochastic blockmodel \cite{abbe16}.
\begin{definition} [\ER Random Graph \cite{bollobas13}]
  Let $n \in \N$ and let $p\in [0,1]$. We recall that the \ER random graph, $\cG(n,p)$, is the
  probability space formed by the graphs defined on the set of vertices $[n]$, where edges are
  drawn randomly from $\begin{pmatrix}n\\2\end{pmatrix}$ independent Bernoulli random variables with
  probability $p$. In effect, a graph $G \sim \cG(n,p)$, with $m$ edges, occurs with
  probability
  \begin{equation}
    \prob{G} = p^{m}(1-p)^{{n\choose 2}-m}.
    \label{eq:3}
  \end{equation}
  \label{ER-def}
\end{definition}
\begin{definition}
  Let $G=(V,E)\sim \cG(N,p)$. For any vertex $u\in V$, we denote by $d_u$ the degree of $u$; we also
  denote by $\dbar = (n-1)p$ the expected value of $d_u$.
  \label{degree-er}
\end{definition}
\noindent We now introduce a model of a dynamic community network: the balanced, two-community
stochastic blockmodel.
\begin{definition}[Dynamic Stochastic Blockmodel]
  Let $n\in \mathbb{N}$, and let $p,q \in [0,1]$.  We denote by $\cG(n,p,q)$ the probability
  space formed by the graphs defined on the set of vertices $[n]$, constructed as follows.\\

  We split the vertices $[n]$ into two communities $C_1$ and $C_2$, formed by the odd and the even
  integers in $[n]$ respectively. We denote by $n_1 = \lfloor (n+1)/2 \rfloor$ and $n_2 = \lfloor
  n/2 \rfloor$  the size of $C_1$ and $C_2$  respectively. \\

  Edges within each community are drawn randomly from independent Bernoulli random variables with
  probability $p$. Edges between communities are drawn randomly from independent Bernoulli random
  variables with probability $q$. For $G \in \cG(n,p,q)$, with $m_1$ and $m_2$ edges in communities
  $C_1$ and $C_2$ respectively, we have
  \begin{equation}
    \prob{G} = p^{m_1}(1-p)^{{m_1 \choose 2}-m_1}q^{m_2}(1-q)^{{m_2 \choose 2}-m_2}.
  \end{equation}
  \label{defn:DSB}
\end{definition}
\begin{remark}
  Although we use $\cG$ for both random graph models, the presence of two or three parameters
  prevents ambiguity in our definitions. 
\end{remark}
\begin{definition}
  Let $G\sim \cG(n,p,q)$. We denote by $\dbar_1 = p n_1$  the expected degree within community $C_1$,
  and by $\dbar_2 = p n_2$  the expected degree within community $C_2$. \\

  \noindent We denote by $k_n$ the binomial random variables that counts the number of
  cross-community edges between $C_1$ and $C_2$.
  \label{degree-sbm}
\end{definition}
\noindent Because asymptotically, $n_1 \sim n_2$, we ignore the dependency of the expected 
degree on the specific community when computing asymptotic behaviors for large $n$. More precisely, we have the following
results.
\begin{lemma}
  Let $G_n \in \cG(n,p,q)$. We have
  \begin{enumerate}
  \item $\dbar_1 = \dbar_2 + \varepsilon p$, where $\varepsilon = 0$ if $n$ is even, or $\varepsilon
    = 1$ otherwise.
  \item $\dbar_1^2 = \dbar_2^2 (1 + \o{1})$.
  \item $\displaystyle \frac{1}{\dbar_1^2} = \frac{1}{\dbar_2^2}\left(1 + \o{1}\right)$.
  \item $\displaystyle \frac{1}{\dbar_1} = \frac{1}{\dbar_2} + \O{\frac{1}{\dbar^2}}$ where  $\dbar =
    \dbar_1$, or $\dbar =  \dbar_2$.\\
  \end{enumerate}
  \noindent In summary, in the remaining of the text we loosely write $1/\dbar$ when either $1/\dbar_1$ or
  $1/\dbar_2$ could be used, and the error between the two terms is no larger than
  $\O{1/\dbar^2}$.
  \label{n1_eq_n2}
\end{lemma}
\begin{remark}
  In this work, we study nested sequences of random graphs, and we use sometime the subscript $n$ to
  denote the index of the corresponding element $G_n$ in the graph process.
\end{remark}
\begin{remark}
  While our model assumes that the two communities have equal size, or differ at most by one
  vertex, the model can be extended to multiple communities of various sizes.
\end{remark}
\begin{figure}[H]
  \centerline{
    \includegraphics[width=0.7\textwidth]{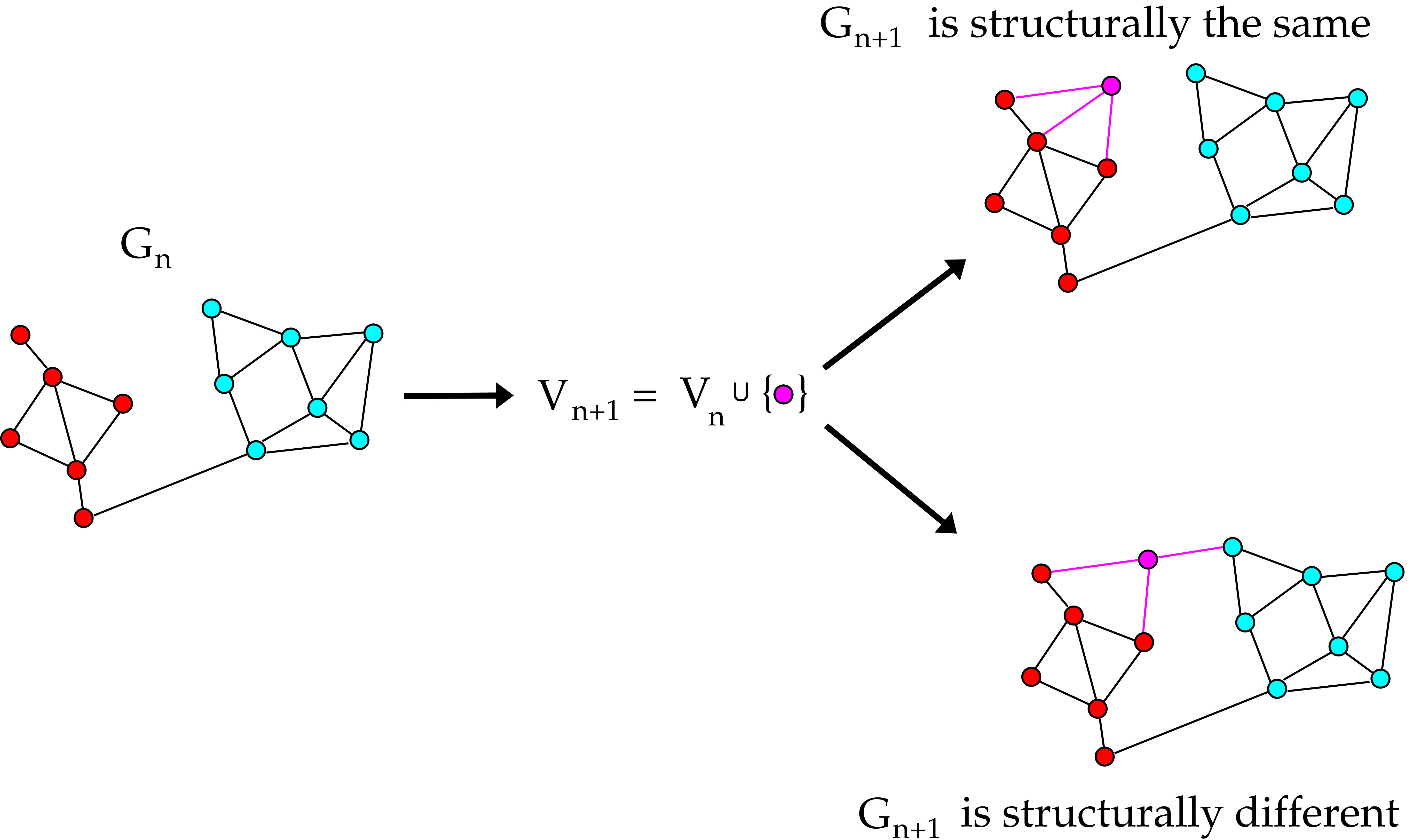}
  }
  \caption{Left: the dynamic stochastic blockmodel $G_n$ is comprised of two communities ($C_1$: red and
    $C_2$: blue). As a new (magenta) vertex is added, the new graph $G_{n+1}$ can remain structurally
    the same -- if no new edges are created between $C_1$ and $C_2$ (top right) -- or can become structurally
    different if the communities start to merge with the addition of new edges between $C_1$ and
    $C_2$ (bottom right).
    \label{fig:sbm}}
\end{figure}
\section{Main Results
  \label{the-main-results}}
\subsection{Informal Presentation of our Results}
Before carefully stating the main result in the next subsection, we provide a back of the envelope
analysis to help understand under what circumstances the resistance metric can detect an anomalous
event in the dynamic growth of a stochastic blockmodel. In particular, we aim to detect whether
cross-community edges are formed at a given timestep. In graphs with few cross-community edges, the
addition of such an edge changes the geometry of the graph significantly. We will show that the
creation of such edges can be detected with high probability when the average in-community degree
dominates the number of cross-community edges.

Figure \ref{fig:sbm} illustrates the statement of the problem. As a new vertex (shown in magenta) is
added to the graph $G_n$, the connectivity between the communities can increase, if edges are added
between $C_1$ and $C_2$, or the communities can remain separated, if no cross-community edges are
created. If the addition of the new vertex promotes the merging of $C_1$ and $C_2$, then we consider
the new graph $G_{n+1}$ to be {\em structurally different} from $G_n$, otherwise $G_{n+1}$ remains
{\em structurally the same} as $G_n$ (see Fig. \ref{fig:sbm}).

The goal of the present work is to detect the fusion of the communities without identifying the
communities. We show that the effective resistance yields a metric that is sensitive to changes in
pattern of connections and connectivity structure between $C_1$ and $C_2$. Therefore it can be
used to detect structural changes between $G_n$ and $G_{n+1}$ without detecting the structure
present in $G_n$.

The informal derivation of our main result relies on the following three ingredients:
\begin{enumerate}
\item each community in $\cG(n,p,q)$ is approximately a ``random graph'' (\ER), $\cG (n/2,p)$;
\item the effective resistance between two vertices $u,v$ within $\cG (n/2,p)$ is 
  concentrated around $2/\dbar = 2/(p(n/2-1)$;
\item the effective resistance between $u \in C_1$ and $v \in C_2$ depends only on the bottleneck
  formed by the $k_n$ cross-community edges, $\er{uv} \approx 1/k_n$.
\end{enumerate}
We now proceed with an informal analysis of the changes in effective resistance distance when
the new vertex, $n+1$, is added to the stochastic blockmodel $G_n$ (see Fig. \ref{fig:sbm}).

We first consider the ``null hypothesis'' where no cross-community edges is formed when vertex $n+1$
is added to the graph. All edges are thus created in the community of $n+1$, say $C_1$ (without any
loss of generality). Roughly $pn/2$ new edges are created, and thus about $\O{n}$ vertices are
affected by the addition of these new edges to $C_1$.

Because the effective resistance between any two vertices $u,v$ in $C_1$ is concentrated
around $2/[p(n_1 -1 +1)] \ge 2/(\dbar +1)$, the changes in resistance after the addition of
vertex $n+1$ is bounded by
\begin{equation}
  \Delta \er{uv} \leq \frac{2}{\dbar}-\frac{2}{\dbar+1} = \O{\frac{1}{\dbar^2}}.
\end{equation}
Although, one would expect that only vertices in community $C_1$ (wherein $n+1$ has been added)
be affected by this change in effective resistance, a more detailed analysis shows that vertices in
$C_2$ slightly  benefit of the increase in connectivity within $C_1$.\\

We now consider the alternate hypothesis, where at least one cross-community edge is formed
after adding $n+1$ (see Fig. \ref{fig:sbm}-bottom right). This additional cross-community edge has
an effect on all pairwise effective resistances. Nevertheless, the most significant perturbation
in $\er{uv}$ occurs for the $n/2 \times n/2$ pairs of vertices in $C_1 \times C_2$. Indeed, if
$u \in C_1$ and $v \in C_2$, the change in effective resistance becomes 
\begin{equation}
  \Delta \er{uv} \approx \frac{1}{k_n}-\frac{1}{k_n+1} = \O{\frac{1}{k_n^2}}.
\end{equation}

In summary, we observe asymptotic separation of the two regimes precisely when
$k_n/\dbar\rightarrow 0$, which occurs with high probability when $n \cdot q_n = \o{p_n}$. We should
therefore be able to use the renormalized resistance distance to test the null hypothesis that no
edge is added between $C_1$ and $C_2$, and that $G_n$ and $G_{n+1}$ are structurally the same.\\

We will now introduce the main character of this work: the dynamic stochastic block model, and we
will then provide a precise statement of the result. In particular, we hope to elucidate our model
of a dynamic community graph, in which at each time step a new vertex joins the graph and forms
connections with previous vertices. The idea of graph growth as a generative mechanism is
commonplace for models such as preferential attachment, but is less often seen in models such as \ER
and the stochastic blockmodel.
\subsection{The Growing Stochastic Blockmodel}
We have described in Definition \ref{defn:DSB} a model for a balanced stochastic block model, where
the probabilities of connections $p$ and $q$ are fixed. However, we are interested in the regime of
large graphs ($n \rightarrow \infty$), where $p$ and $q$ cannot remain constant. In fact the
probabilities of connection, within each community and across communities go to zero as the size of
the graph, $n$, goes to infinity.

The elementary growth step, which transforms $G_n=(V_n,E_n)$ into $G_{n+1}=(E_{n+1},V_{n+1})$
proceeds as follows: one adds a vertex $n+1$ to $V_n$ to form $V_{n+1}$, assigns this new vertex to
$C_1$ or $C_2$ according to the parity of $n$. One then connects $n+1$ to each member of its
community with probability $p$ and each member of the opposite community with probability $q$.
This leads to a new set of vertices, $E_{n+1}$.
\begin{table}[H]
  \begin{center}
    \begin{small}
      \begin{tabular}{cccccccc}
        \toprule
        probabilities & \multicolumn{4}{c}{growth sequence}    & &  definition of $G_n$& definition of $D_n$\\
        of connection & \multicolumn{4}{c}{to generate $G_n$}\\
        \midrule
        $\{p_1,q_1\}$ & $G^{(1)}_1$ & \fbox{$G^{(1)}_2$} &            & & & $G_2 \eqdef G^{(1)}_2$ & $D_1 \eqdef \RD(G^{(1)}_1,G^{(1)}_2)$\\
        $\{p_2,q_2\}$ & $G^{(2)}_1$ & $G^{(2)}_2$ & \fbox{$G^{(2)}_3$} & & & $G_3 \eqdef G^{(2)}_3$ & $D_2 \eqdef \RD(G^{(2)}_2,G^{(2)}_3)$\\
        $\{p_3,q_3\}$ & $G^{(3)}_1$ & $G^{(3)}_2$ & $G^{(3)}_3$ & \fbox{$G^{(3)}_4$} & & $G_4 \eqdef G^{(3)}_4$& $D_3 \eqdef \RD(G^{(3)}_3,G^{(3)}_4)$\\
        \vdots & \vdots & \vdots & \vdots & \vdots & \vdots & \vdots & \vdots \\
        \bottomrule
      \end{tabular}
    \end{small}
  \end{center}
  \caption{Each row depicts the growth sequence that leads to the construction of $G_{n+1}\eqdef
    G^{(n)}_{n+1}$. The distance $D_n$  is always defined with respect to the subgraph $G^{(n)}_n$ on the vertices
    $1,\ldots, n$ that led to the construction of $G_{n+1}$.
    \label{growth}
  }
\end{table}
The actual sequence of graphs $\left\{G_n\right\}$ is created using this elementary process with a
twist: for each index $n$, the graph $G_{n+1}$ is created by iterating the elementary growth process,
starting with a single vertex and no edges, $n+1$ times with the {\em fixed probabilities of connections}
$p_n$ and $q_n$. Once $G_{n+1}$ is created, the growth is stopped, the probabilities of connections
are updated and become $p_{n+1}$ and $q_{n+1}$. A new sequence of graphs is initialized to create
$G_{n+2}$.

Table \ref{growth} illustrates the different sequences of growth, of increasing lengths,
that lead to the creation of $G_1, G_2,\ldots $. This growth process guarantees that $G_{n+1}$ is
always a subgraph of $G_n$, and that both $G_n$ and $G_{n+1}$ have been created with the same
probabilities. Furthermore, each $G_n$ is distributed according to Definition \ref{defn:DSB}, and
the $G_n$ are independent of one another.

In order to study the dynamic evolution of the graph sequence, we focus on changes between two
successive time steps $n$ and $n+1$. These changes are formulated in the form of the distance
$D_n = \RD(G^{(n)}_n,G_{n+1})$ between $G_{n+1}$ and the subgraph $G^{(n)}_n$ on the vertices    $1,\ldots, n$,
which led to the construction of $G_{n+1}$. The subgraph $G^{(n)}_n$ is the graph on the left of the
boxed graph $G_{n+1}$ on each row of Table \ref{growth}. The definition of $D_n$ is the only
potential caveat  of the model: $D_n$ is not the distance between $G_n$ and $G_{n+1}$; this
restriction is necessary since in general $G_n$ is not a subgraph of $G_{n+1}$.\\

This model provides a realistic prototype for the separation of scales present in the dynamics of large
social network. Specifically, the time index $n$ corresponds to the slow dynamics associated with
the evolution of the networks over long time scale (months to years). In contrast, the random
realizations on each row of Table  \ref{growth} embody the fast random fluctuations of the
network over short time scales (minutes to hours).

In this work, we are interested in examining fluctuations over fast time scales (minutes to
hours). We expect that the probabilities of connection, $(p_n,q_n)$ remain the same when we study
the distance between $G_n$ and $G_{n+1}$. As $n$ increases, the connectivity patterns of members of
the network evolve, and we change accordingly the probabilities of connection, $(p_n,q_n)$. Similar
dynamic stochastic block models have been proposed in the recent years (e.g.,
\cite{ho15,pensky16,tang14,wilson16,xing10,xu14,xu15,yang11}, and references therein).

In the stochastic blockmodel, each vertex $u$ belongs to a community within the graph. If vertex $u$
forms no cross-community edges, then the geometry of the graph is structurally the same. However, if
$u$ forms at least one cross-community edge, then (depending on the geometry of the preceding graphs
in the sequence) the geometry may change significantly. We examine in what regimes of $p_n$ and
$q_n$ we can differentiate between the two situations with high probability.  We phrase the result
in terms of a hypothesis test, with the null hypothesis being that no cross-community edges have
been formed in step $n+1$.
\begin{figure}[H]
  \centering
  \includegraphics[width=0.8\textwidth]{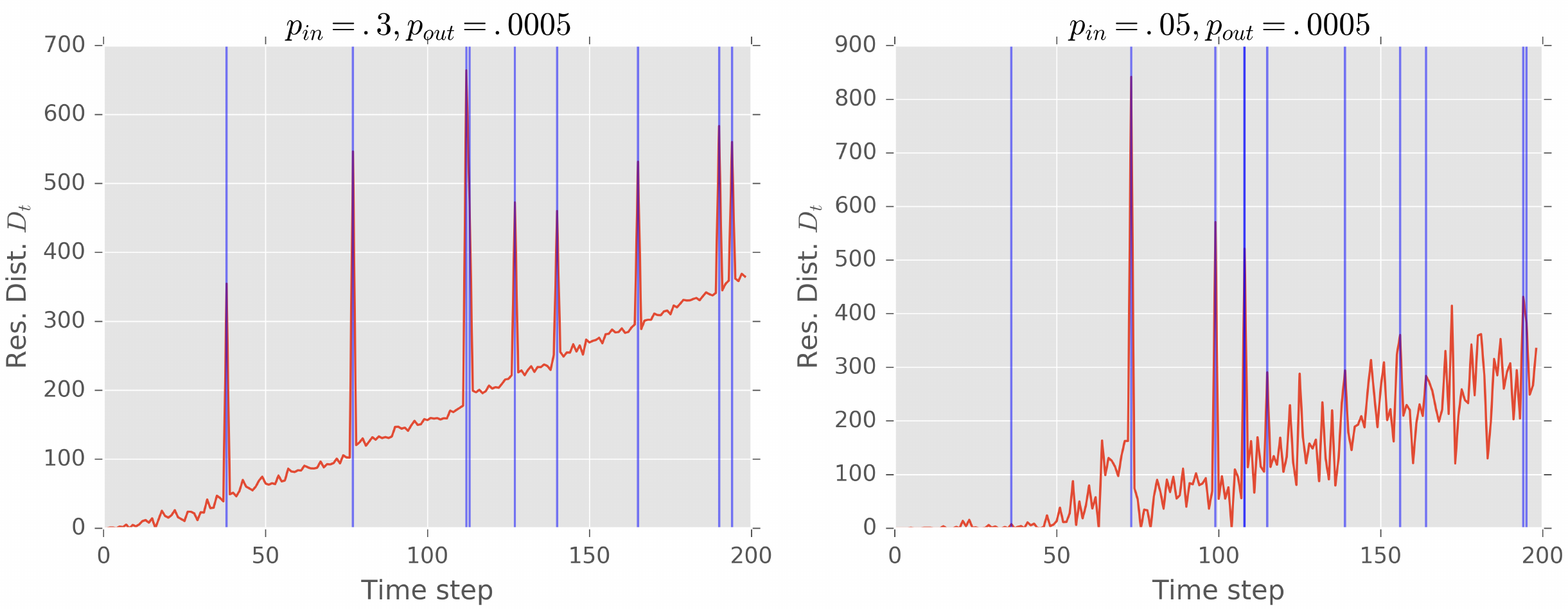}
  \caption{A typical time series of $D_n$ for a growing stochastic blockmodel. The red curve is the
    distance between time steps $D_n$ and the blue vertical lines mark the formation of
    cross-community connections. Two different regimes are compared. On the left, the formation of
    cross-community edges is easily discernible, while on the right, such an event is quickly lost
    in the noise.}
  \label{fig:timeseries}
\end{figure}

Figure \ref{fig:timeseries} shows a time series of distances $D_n$ for a growing stochastic
blockmodel. We see that when the in-community connectivity is much greater than the cross community
connectivity, the formation of cross-community edges is easily discernible (left figure). However,
when the level of connectivity is insufficiently separated, then the formation of cross-community
edges is quickly lost in the noise. Our result clarifies exactly what is meant when we say that the
parameters $p_n$ and $q_n$ are ``well separated.''\\

\noindent Our main result is given by the following theorem.
\begin{Theorem}
  \label{main_theorem}
  Let $G_{n+1}\sim\cG(n+1,p_n,q_n)$ be a stochastic blockmodel with $p_n=\om{\log n/n}$ ,
  $q_n = \om{1/n^2}$, $q_n = \o{p_n/n}$, and $p_n = \O{1/\sqrt{n}}$. Let $G_{n}$ be the subgraph
  induced by the vertex set $[n]$, with $m_n$ edges. Let $D_n = \RD\left(G_n,G_{n+1}\right)$ be the
  normalized effective resistance distance, $\RD$, defined in
  \eqref{eq:metric-defn}.\\

  \noindent To test the hypothesis
  \begin{equation}
    H_0: \quad k_n = k_{n+1}
  \end{equation}
  versus 
  \begin{equation}
    H_1: \quad k_n < k_{n+1}
  \end{equation}
  we use the test based on the statistic $Z_n$ defined by
  \begin{equation}
    Z_n \eqdef \frac{16m_n^2}{n^4} \left(D_n - n \right),
    \label{def_Z_n}
  \end{equation}
  where we accept
  $H_0$ if $Z_n < z_\varepsilon$ and accept $H_1$ otherwise. The threshold $z_\varepsilon$ for the rejection
  region satisfies
  \begin{equation}
    \proba_{H_0}\left(Z_n \ge z_\varepsilon\right) \leq \varepsilon \quad \text{as}\quad n \rightarrow \infty,
  \end{equation}
  and
  \begin{equation}
    \proba_{H_1}\left(Z_n \ge z_\varepsilon\right) \rightarrow 1 \quad \text{as}\quad n \rightarrow \infty.
  \end{equation}
  The test has therefore asymptotic level $\varepsilon$ and asymptotic power 1.
\end{Theorem}
\begin{proof}
  The proof of the theorem can be found in Appendix \ref{app:proofs}.
\end{proof}
\begin{figure}[H]
  \centering
  \includegraphics[width=\textwidth]{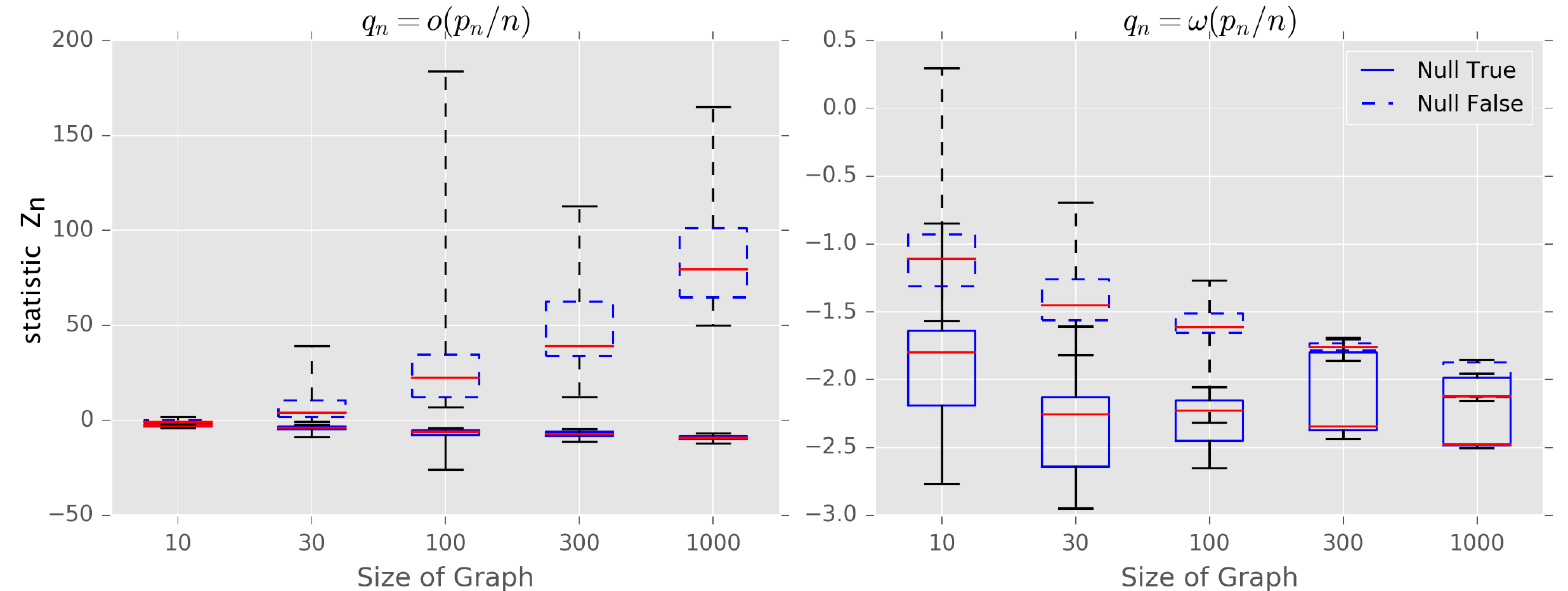}
  \caption{Empirical distribution of $Z_n$ under the null hypothesis (solid line) and the alternate
    hypothesis (dashed line). The data is normalized so that the distribution has zero mean and unit
    variance under the null hypothesis.  The probability of connection within each community is
    $p_n = \log^2n / n$.  The probability of connection between $C_1$ and $C_2$ is $q_n=\log n /n^2$
    (left) and $q_n = \log^2n / n^{3/2}$ (right). The box extends from the lower to upper quartile
    values of the data, with a line at the median. The whiskers extend from the box to show the full
    range of the data. Note that the y-axis is logarithmic in the left figure and linear in the
    right.}
  \label{fig:separation}
\end{figure}
\begin{remark}
  In practice, it would be desirable to have an analytical expression for the constant
  $z_\varepsilon$ such that we can compute a level $\varepsilon$ test,
  \begin{equation*}
    \prob{Z_n\geq z_\varepsilon} \leq \varepsilon \quad \text{under the
      null hypothesis.}
  \end{equation*}
  Unfortunately, our technique of proof, which is based on the asymptotic behavior of
  $Z_n$ does not yield such a constant. A more involved analysis, based on finite sample
  estimates of the distance, would be needed, and would yield an important extension of the present
  work. The results shown in Figure~\ref{fig:separation} suggest that one could numerically estimate a
  $1-\varepsilon$ point wise confidence interval for $Z_n$ with a bootstrapping technique;
  the details of such a construction are the subject of ongoing investigation.
\end{remark}
\section{Experimental Analysis of Dynamic Community Networks
  \label{experiments}}
Figure \ref{fig:separation} shows numerical evidence supporting Theorem \ref{main_theorem}. The
empirical distribution of $Z_n$ is computed under the null hypothesis (solid line) and the alternate
hypothesis (dashed line). The data are scaled so that the empirical distribution of $Z_n$ under
$H_0$ has zero mean and unit variance.  In the left and right figures, the density of edges remains
the same within each community, $p_n = \log^2n / n$.

The plot on the left of Fig. \ref{fig:separation} illustrates a case where the density of
cross-community edges remains sufficiently low -- $q_n=\log n /n^2$ -- and the test statistic can
detect the creation of novel cross-community edges (alternate hypothesis) without the knowledge of
$k_n$, or the identification of the communities.

On the right, the density of cross-community edges is too large -- $q_n = \log^2n / n^{3/2}$ -- for
the statistic to be able to detect the creation of novel cross-community edges. In that case the
hypotheses of Theorem \ref{main_theorem} are no longer satisfied.

In addition to the separation of the distributions of $Z_n$ under $H_0$ and $H_1$ guaranteed by
Theorem \ref{main_theorem} when $q_n / p_n = \o{1/n}$, we start observing a separation between the
two distributions when $q_n / p_n = \T{1/n}$ (not shown) suggesting that the hypotheses of Theorem
\ref{main_theorem} are probably optimal. In the regime where $q_n = \om{p_n/n}$, shown in
Fig. \ref{fig:separation}-right, the two distributions overlap.

\section{Discussion
  \label{discussion}}
At first glance, our result may seem restrictive compared to existing results regarding community
detection in the stochastic blockmodel. However, such a comparison is ill-advised, as we do not
propose this scheme as a method for community detection. For example, Abbe et al. have shown that
communities can be recovered asymptotically almost surely when $p_n = a \log(n) / n$ and
$q_n = b \log(n) / n$, provided that $(a+b)/2 - \sqrt{ab} > 1$ \cite{abbe16}. Their method uses an
algorithm that is designed specifically for the purpose of community detection, whereas our work
provides a very general tool, which can be applied on a broad range of dynamic graphs, albeit
without the theoretical guarantees that we derive for the dynamic stochastic block model.

Furthermore, the ``efficient'' algorithm proposed by Abbe et al. is only proven to be
polynomial time, whereas resistance matrices can be computed in near-linear or quadratic time,
for the approximate \cite{Spielman2011} and exact effective resistance respectively. This allows our
tool to be of immediate practical use, whereas results such as those found in \cite{abbe16}
are of a more theoretical flavor.

Some argue against the use of the effective resistance to analyze connectivity properties of a
graph. In \cite{luxburg14} it is shown that
\begin{equation}
  \label{eq:29}
  \left| \hR_{u\, v} - \frac{1}{2}\left( \frac{1}{d_u} + \frac{1}{d_v} \right)\right| \leq
  \left(\frac{1}{1-\lambda_2} + 2\right)\frac{w_{max}}{\delta_n^2}, 
\end{equation}
where $\lambda_2$ is the second largest eigenvalue of the normalized graph Laplacian, $d_v$ is the
degree of vertex $v$, $\delta_n$ is the minimum degree, and $w_{max}$ is the maximum edge
weight. If the right-hand side of \eqref{eq:29} converges to 0, then the effective resistance will
converge to the average inverse degree of $u$ and $v$. 

Luxburg et al. argue that the result \eqref{eq:29} implies that when the bound converges to zero,
the resistance will be uninformative, since it depends on local properties of the vertices and not
global properties of the graph. Fortunately, this convergence can coexist peacefully alongside our
result. In particular, if both expected degree and minimum degree approach infinity with high
probability, as they will when both $q_n$ and $p_n$ are $\om{1/n^2}$, then such convergence will
itself occur with high probability (see (\ref{spectral-gap}) for the relevant spectral gap bound). Since we
only care about relative changes in resistance between $G_n$ and $G_{n+1}$ in cases where
cross-community edges are and are not formed, this is no problem for us. That said, the warning put
forth by Luxburg et al. is well taken; we must be careful to make sure that we understand the
expected behavior of the distance $\RD(G_n,G_{n+1})$ and compare the observed behavior to this
expected behavior rather than evaluate it on an absolute scale, since in many situations of interest
this distance will converge to zero as the graph grows.

Luxburg et al. have pointed out that the resistance can be fickle when used on graphs with high
connectivity, which is to say a small spectral bound $(1-\lambda_2)^{-1}$. We now know in which
circumstances this will become an issue when looking at simple community structure. Further
investigation is needed to know when other random graph models such as the small-world or
preferential attachment model will be susceptible to analysis via the renormalized resistance
metric. We are currently investigating the application of this analysis to a variety of real-world
data sets.

\appendix
\section{Asymptotic Notations}\label{appx:asymptotic}
If $\{a_n\}_{n=1}^\infty$ and $\{b_n\}_{n=1}^\infty$ are
infinite sequences. The notations on the left have the interpretation on the right,
\begin{align*}
  a_n & = \O{b_n}   &  \exists n_0 > 0, \exists c > 0,& \forall n \ge n_0, \quad 0 \le |a_n|\le c |b_n|,\\
  a_n & = \o{b_n}   &  \forall c > 0, \exists n_0\ge 0, &  \forall n \ge n_0, \quad 0 \le |a_n| \le c |b_n|,\\
  a_n & = \om{b_n}   &  \forall c > 0, \exists n_0 \ge 0, & \forall n \ge n_0,  \quad 0 \le c|b_n| < |a_n|,\\
  a_n & = \T{b_n}  & \exists c_1, c_2 >0, \exists n_0 \ge 0,& \forall n \ge n_0,  \quad 0 \le c_1 |b_n| \leq |a_n| \leq
                                                              c_2 |b_n|.
\end{align*}
We can adapt any of the above statements to doubly-indexed sequences $a_{n,k}$ and $b_{n,k}$ by
requiring that there exist an $n_0\geq0$ such that the conditions on the right hold for all
$n,k\geq n_0$.
\section{Notation}	
\label{the-notations}
\
\begin{center}
  \begin{small}
    \begin{tabular}{llr}
      \toprule
      \sf Symbol & \sf Definition & \sf Definition or \\
                 &                & \sf Equation Number \\
      \midrule
      $[n]$ & The subset of natural numbers $\left\{1,\ldots,n\right\}$.\\
      $\cG(n,p)$ & \ER random graph with parameters $n$ and $p$ & \ref{ER-def}\\
      $\cG(n,p,q)$ & stochastic blockmodel with parameters $n$, $p$, and $q$ & \ref{defn:DSB}\\
      $G_{n}$ & Subgraph of a graph $G$ induced by the vertex set $[n]$ \\
      $\er{uv}$ & Effective resistance between $u$ and $v$ & \ref{eff-res}\\
      $R_{u\, v}$ & Renormalized effective resistance between $u$ and $v$ & \eqref{eq:renorm-res-defn} \\
      $R^{(n)}_{u\, v}$ & Renormalized effective resistance between $u$ and $v$ in $G_n$ \\
      $d_u$ & Degree of vertex $u$ (random variable) & \ref{degree-er} \& \ref{degree-sbm}\\
      $\dbar$ & Mean degree (Erd\H{o}s-R\'enyi) or expected in-community degree (stochastic blockmodel) \\
      $k_n$ & Number of cross-community edges (random variable) & \ref{degree-sbm}\\
      $\E{k_n}$ & Expectation of the number of cross-community edges \\
      $m_n$ & Total number of edges\\
      $\RP(\cdot,\cdot)$ & Resistance-perturbation distance & (\ref{rp_dist_def_eqn})\\
      $\RD(\cdot,\cdot)$ & Renormalized resistance distance (with $\beta=1$) & \ref{defn:metric}\\
      \bottomrule
      \label{tab:notations}
    \end{tabular}
  \end{small}
\end{center}
\section{Proof of Main Result
  \label{app:proofs}}
We begin by proving a lemma that allows us to transfer bounds on changes in effective resistance
into bounds on changes in renormalized resistances.
\begin{lemma}
  \label{lemma:change_bounds}
  Suppose that $\hR_1$ and $\hR_2$ are two effective resistances. If
  \begin{equation*}
    \label{eq:41}
    C_1 \leq \left|\hR_1 - \hR_2 \right| \leq C_2,
  \end{equation*}
  then the corresponding \emph{renormalized} resistances obey
  \begin{equation}
    \label{eq:42}
    \frac{C_1}{(\hR_1+1)(\hR_2+1)} \leq \left| R_1 - R_2 \right| \leq C_2.
  \end{equation}
\end{lemma}
\begin{proof}
  Recall that the renormalized resistance corresponding to $\hR$ is given by $R = f(\hR)$
  where $f(x) = x/(x+1)$. The mean value theorem thus implies that
  \begin{equation*}
    \label{eq:14}
    \left| R_{1} - R_{1}\right| \leq \sup_{x\in \mathbb{R}}|f'(x)| \left| \hR_{1} - 
      \hR_{2}\right| \leq  \left| \hR_{1} - \hR_{2}\right|\leq C_2. 
  \end{equation*}
  To obtain the lower bound, we compute:
  \begin{equation}
    \left| R_{1} - R_{1}\right| = \left|\frac{\hR_1}{\hR_1+1} - \frac{\hR_2}{\hR_2 + 1}\right| 
    = \frac{|\hR_1 - \hR_2|}{(\hR_1+1)(\hR_2+1)} 
    \ge \frac{C_1}{(\hR_1+1)(\hR_2+1)}.
  \end{equation}
  In our calculation above, we used the fact that $\hR_1$ and $\hR_2$ are non negative.
\end{proof}
\subsection{Resistance Deviations in \ER}
We begin by analyzing the perturbations of the distance $\RD(G_n,G_{n+1})$, defined by
\eqref{eq:RD}, when $G_n\sim \cG(n,p_n)$ is an \ER random graph. Our ultimate goal is to understand
a stochastic blockmodel, and we will leverage our subsequent understanding of the \ER model to help
us in achieving this goal.
\begin{lemma}
  \label{lemma:HPbound1}
  Let $G\sim \cG(n,p_n)$ be fully connected, with $p_n = \om{\log n/n}$. For any two vertices
  $u$,$v$ in $G$, we have
  \begin{equation} 
    \left| \er{uv} - \left(\frac{1}{d_u}+\frac{1}{d_v}\right)\right| =
    \O{\frac{1}{\dbar^2}} \quad \text{with high probability.}
    \label{R-ER}
  \end{equation}
\end{lemma}
\begin{remark}
  The authors in \cite{luxburg14} derive a slightly weaker bound,
  \begin{equation} 
    \left| \er{uv} - \left(\frac{1}{d_u}+\frac{1}{d_v}\right)\right| =
    \o{\frac{1}{\delta_n}}. 
  \end{equation}
  We need the tighter factor $\O{1/\dbar^2}$; and thus we derive the bound (\ref{R-ER}) using one of
  the key results  (Proposition 5) in \cite{luxburg14}.
\end{remark}
\begin{proof}
  Define $\bD$ to be the diagonal matrix with entries $d_1,\ldots,d_n$. Since all degrees are
  positive ($G$ is fully connected with high probability), we denote by $\bD^{-1/2}$ the diagonal matrix with entries
  $1/\sqrt{d_1},\ldots,1/\sqrt{d_n}$.  Let $\bA$ be the adjacency matrix of $G$.  Define
  $\bB = \bD^{-1/2}\bA\bD^{-1/2}$, with eigenvalues $1=\lambda_1\ge \lambda_2 \ge... \ge \lambda_n$.
  As explained above, we use Proposition 5 in \cite{luxburg14} to bound the deviation of $\er{uv}$ away from
  $1/d_u+ 1/d_v$,
  \begin{equation}
    \label{eq:5}
    \left| \er{uv} - \left(\frac{1}{d_u}+\frac{1}{d_v}\right)\right| \leq
    \frac{2}{\delta_n^2}\left(2+\frac{1}{1-\lambda_2}\right), 
  \end{equation}
  where $\delta_n$ is the minimum degree. Define 
  \begin{equation}
    \varepsilon_n= \sqrt{6 \log(n) \dbar}.
  \end{equation}
  We apply Chernoff's bound on  the degree distribution, 
  \begin{equation*}
    \prob{\left\lvert d_v - \dbar \right\rvert \ge \varepsilon_n} \le
    2 \exp\left(-\frac{\varepsilon_n^2}{3\dbar}\right).
  \end{equation*}
  Now, 
  \begin{equation}
    \frac{\varepsilon_n^2}{3\dbar} = \frac{6 \dbar \log n}{3\dbar} = \log n^2,
  \end{equation}
  and thus
  \begin{equation*}
    \prob{\left\lvert d_v - \dbar \right\rvert \ge \varepsilon_n}\le \frac{2}{n^2}.
  \end{equation*}
  In the end,  applying a union bound on all $n$ vertices yields
  \begin{equation}
    \prob{\forall v \in [n], \left\lvert d_v - \dbar \right\rvert \ge \varepsilon_n} \le \frac{2}{n}.
  \end{equation}
  We consider $d_v$ in the interval $[\dbar - \varepsilon_n, \dbar + \varepsilon_n]$.  The mean value
  theorem implies that there exists $\tilde{d} \in (\dbar,d_v)$, or $\tilde{d} \in (d_v,\dbar)$, such
  that
  \begin{equation}
    \left\lvert \frac{1}{d_v^2} - \frac{1}{\dbar^2}\right\rvert = 
    \frac{2}{\tilde{d}^3} \left\lvert \dbar - d_v\right\rvert. 
  \end{equation}
  Now,
  \begin{equation}
    \frac{2}{\tilde{d}^3} \left\lvert \dbar - d\right\rvert 
    \le \frac{2\varepsilon_n}{(\dbar -\varepsilon_n)^3}
    = \frac{2\varepsilon_n}{\dbar^3}\frac{1}{\left(1 - \varepsilon_n/\dbar\right)^3}.
  \end{equation}
  At last, we use the following elementary fact
  \begin{equation}
    0 \le \frac{1}{(1 - x)^3} \le 1 + 12 x, \quad \text{if}\, x < 1/4,
  \end{equation}
  to conclude that 
  \begin{equation}
    \forall v \in [n], \quad
    \left\lvert \frac{1}{d_v^2} - \frac{1}{\dbar^2}\right\rvert \le
    \frac{2\varepsilon_n}{\dbar^3}\left( 1 + 12 \frac{\varepsilon_n}{\dbar}\right)
    \quad \text{with probability greater than $2/n$}.
  \end{equation}
  This eventually yields an upper bound on the inverse of the minimum degree squared,
  \begin{equation}
    \frac{1}{\delta_n^2} \le \frac{1}{\dbar^2} + 
    \frac{2\varepsilon_n}{\dbar^3}\left( 1 + 12 \frac{\varepsilon_n}{\dbar}\right)
    \quad \text{with probability greater than $2/n$}.
    \label{one-over-dbar}
  \end{equation}
  To complete the proof of the lemma, we use a lower bound on the spectral gap $1-\lambda_2$. Because
  the density of edges is only growing faster than $\log n/n$, we use the optimal bounds given by
  \cite{coja-oghlan07}. Applied to the eigenvalue $\lambda_2$ of $\bB$, Theorem 1.2 of
  \cite{coja-oghlan07} implies that
  \begin{Theorem}[\cite{coja-oghlan07}]
    If $\dbar > c \log n/n$, then with high probability,
    \begin{equation}
      1-\frac{c}{\sqrt{\dbar}} \leq 1- \lambda_2 \leq 1+\frac{c}{\sqrt{\dbar}}. 
      \label{spectral-gap}
    \end{equation}
  \end{Theorem}
  \noindent The lower bound in (\ref{spectral-gap}) yields the following upper bound, with high probability,
  \begin{equation}
    \frac{1}{1- \lambda_2} \leq 1+\frac{c}{\sqrt{\dbar}}. 
    \label{one-over-lambda2}
  \end{equation}
  Using the bounds given by (\ref{one-over-dbar}) with (\ref{one-over-lambda2}), which happen both
  with high probability, in \eqref{eq:5} yields the advertised result. 
\end{proof}
An important corollary of lemma \ref{lemma:HPbound1} is the concentration of $\er{uv}$ around
$2/\dbar$ (see also \cite{luxburg14} for similar results),
\begin{corollary}
  Let $G\sim \cG(n,p_n)$ be fully connected, with $p_n = \om{\log n/n}$. With high probability,
  \begin{equation} 
    \left| \er{uv} - \frac{2}{\dbar}\right| =
    \frac{16}{\dbar^2} + \o{\frac{1}{\dbar^2}}.
    \label{reff-dmean}
  \end{equation}
  \label{concentration}
\end{corollary}
The lemma provides a confidence interval for the effective resistance $\er{uv}$ centered at
$d_u^{-1}+d_v^{-1}$, for pairs of vertices present in the graph $G$ at time $n$. We now use this
result to bound the change in (renormalized) resistance between a pair of vertices $u$ and $v$
present in $G_n$, when the graph grows from $G_n$ to $G_{n+1}$.
\begin{Theorem}
  \label{thm:DRbound}
  Let $G_{n+1}\sim\cG(n+1,p_n)$ be an \ER random graph with $p_n = \om{\log n/n}$. Let $G_n$ be the
  subgraph induced by the vertices $[n]$ in $G_{n+1}$, and let $\dbar = (n-1)p_n$ be the expected
  degree in $G_n$.\\

  \noindent The change in renormalized effective resistance, when the graph $G_n$ becomes $G_{n+1}$,  is given by 
  \begin{equation}
    \label{eq:9}
    \max_{u<v\leq n} \left| R^{(n+1)}_{u\, v} - R^{(n)}_{u\, v}\right| = \O{\frac{1}{\dbar^2}} \quad
    \text{with high probability.}
  \end{equation}
\end{Theorem}
\begin{remark}
  It is important to note that the bound on changes in $R_{u\, v}$ from time $n$ to $n+1$ only holds for
  the nodes $u,v \in [n]$ that are already present in $\cG(n,p_n)$. Indeed, for the new node $n+1$
  that is added at time $n+1$, we have $\hR^{(n)}_{u\, n+1} = \infty$, and thus $R^{(n)}_{u\, n+1} =
  1$. In this case, the bound in \eqref{eq:9} is replaced by
  \begin{equation}
    \left| R^{(n+1)}_{u\, n+1} - R^{(n)}_{u\, n+1}\right| 
    =   R^{(n)}_{u\, n+1} - R^{(n+1)}_{u\, n+1}
    = 1 - R^{(n+1)}_{u\, n+1} \leq 1\quad \text{for  any vertex $u \in [n]$}.
    \label{Rn-plus-one}
  \end{equation}
\end{remark}
\begin{proof} 
  By Lemma \ref{lemma:change_bounds}, it suffices to prove the inequality with respect to the
  effective resistance, rather than the renormalized resistance.

  Let $d_u^{(n)}$ denote the degree of vertex $u$ in $G_n$, and similarly define $d_u^{(n+1)}$. Using the triangle inequality,
  \begin{align}
    \left| \hR^{(n+1)}_{u\, v} - \hR^{(n)}_{u\, v}\right| 
    &\leq 
      \left| \hR^{(n)}_{u\, v} -\left(\frac{1}{d_u^{(n)}}+\frac{1}{d_v^{(n)}}\right)\right|
      + \left| \left(\frac{1}{d_u^{(n)}}+\frac{1}{d_v^{(n)}}\right)
      - \left(\frac{1}{d_u^{(n+1)}}+\frac{1}{d_v^{(n+1)}}\right)\right| \label{eq:15}\\
    &\quad \quad + \left| \hR^{(n+1)}_{u\, v} -\left(\frac{1}{d_u^{(n+1)}}+\frac{1}{d_v^{(n+1)}}\right)\right|.\nonumber
  \end{align}
  From lemma \ref{lemma:HPbound1}, we obtain bounds on the first and third terms, of order
  $\O{1/\delta_n^{2}}$.  Since $|d_u^{(n+1)}-d_u^{(n)}| \leq 1$, the middle term is of order
  $\O{1/\delta_n^{2}}$, which can in turn be bounded by a term of order $\O{1/\dbar^2}$ with
  high probability using \ref{one-over-dbar}. Putting everything together, we get
  \begin{equation*}
    \label{eq:1}
    \left| \hR^{(n+1)}_{u\, v} - \hR^{(n)}_{u\, v}\right|=\O{\frac{1}{\dbar^2}}.
  \end{equation*}
  The inequality is proven for the effective resistance, and using   lemma \ref{lemma:change_bounds}
  it also holds for the renormalized resistance.
\end{proof}
\subsection{Effective resistances in the stochastic blockmodel}
We first recall that the number of cross-community edges, $k_n$, is a binomial distribution, and thus
concentrates around its expectation $\E{k_n}$ for large $n$, as explained in the following lemma.
\begin{lemma}
  \label{k_n_concentrates}
  Let $G_{n}\sim\cG(n,p_n,q_n)$ be a (balanced two-community) stochastic blockmodel with
  $p_n=\om{\log n/n}$ and $q_n=\om{1/n^2}$. There exists $n_0$, such that 
  \begin{equation}
    \forall n \ge n_0, \quad \frac{3}{4} < \frac{\E{k_n}}{k_n} <  \frac{3}{2},\quad \text{with probability > 0.9}.
  \end{equation}
\end{lemma}
\begin{proof}
  The random variable $k_n$ is binomial $B(n_1n_2,q_n)$, where
  $n_1 = \lfloor (n+1)/2 \rfloor$, and $n_2 = \lfloor n/2 \rfloor$. We have $\E{k_n} = n_1n_2
  q_n$. We apply a Chernoff's bound on $k_n$ to get
  \begin{equation}
    \prob{|k_n - \E{k_n}| > \varepsilon} < 2 e^{-\varepsilon^2/(3\E{k_n})}.
  \end{equation}
  Using $\varepsilon = 3\sqrt{\E{k_n}}$, we get
  \begin{equation}
    \prob{|k_n - \E{k_n}| > 3 \sqrt{\E{k_n}}} < 2 e^{-3} < 0.1
  \end{equation}
  or 
  \begin{equation}
    1 - 3\frac{1}{\sqrt{\E{k_n}}} < \frac{k_n}{\E{k_n}} <  1 + 3\frac{1}{\sqrt{\E{k_n}}} \quad \text{with probability > 0.9}.
  \end{equation}
  Now, $\E{k_n} = q_n n_1n_2 = q_n\O{n^2} = \om{1}$, and thus $\lim_{n \rightarrow \infty} \E{k_n} =
  \infty$. Consequently  $\exists n_0$ such that 
  \begin{equation}
    \forall n \ge n_0, \quad \E{k_n} > 81,
  \end{equation}
  and thus 
  \begin{equation}
    \forall n \ge n_0, \quad \frac{3}{4} < \frac{\E{k_n}}{k_n} <  \frac{3}{2},\quad \text{with probability > 0.9}.
  \end{equation}
\end{proof}
We now translate our understanding of the \ER random graph to the analysis of the stochastic
blockmodel.  As explained in lemma \ref{n1_eq_n2}, in the following we write $1/\dbar$ when either
$1/\dbar_1$ or $1/\dbar_2$ could be used, and the error between the two terms is no larger than
$\O{1/\dbar^2}$.
\begin{lemma}[Cross-community resistance bounds]
  \label{lemma:cross-comm}
  Let $G_{n}\sim\cG(n,p_n,q_n)$ be a (balanced two-community) stochastic blockmodel with
  $p_n=\om{\log n/n}$ and $q_n=\om{1/n^2}$. Let $u$ and $v$ be vertices in the communities $C_1$ and
  $C_2$ respectively. Let $\dbar$ be the expected in-community degree of $C_1$. Let $k_n$ be the
  (random) number of cross-community edges. With high probability, the effective resistance
  $\er{uv}$ is bounded according to
  \begin{equation}
    \frac{1}{k_n} \leq \hR_{u\, v} \leq \frac{1}{k_n} + \frac{4}{\dbar} +
    \O{\frac{1}{\dbar^2}}.
    \label{lax-bound}
  \end{equation}
\end{lemma}
\begin{remark}
  We recall (see lemma \ref{n1_eq_n2}) that when we write $\dbar$, in \ref{eq:theorem4}, it either
  means the expected degree of $C_1$ or $C_2$.
\end{remark}
%
\begin{remark}
  The requirement that $p_n=\omega(\log(n)/n)$ guarantees that we are in a regime where resistances
  in the \ER graph converge to $2/\dbar$. The requirement that $q_n=\omega(1/n^2)$ guarantees that
  $\E{k_n}\rightarrow\infty$, and because of lemma \ref{k_n_concentrates}, $k_n\rightarrow\infty$
  with high probability. Finally, $q_n = o(p_n/n)$ guarantees that $\E{k_n}= \o{\dbar}$, and using
  lemma \ref{k_n_concentrates} we have $k_n= \o{\dbar}$  with high probability.
\end{remark}
\begin{proof}[Proof of lemma \ref{lemma:cross-comm}]
  \noindent Without loss of generality, we assume that $u \in C_1$, and $v \in C_2$ (see
  Fig. \ref{nash-williams-fig}). To obtain the lower bound on $\er{uv}$ we use the Nash-Williams
  inequality \cite{lyons11}, which we briefly recall here. Let $u$ and $v$ be two distinct
  vertices. A set of edges $E_c$ is an {\em edge-cutset} separating $u$ and $u$ if every path from
  $u$ to $v$ includes an edge in $E_c$.
  \begin{lemma}[Nash-Williams, \cite{lyons11}]
    If $u$ and $v$ are  separated by $K$ disjoint edge-cutsets $E_k, k =1,\ldots, K$, then
    \begin{equation}
      \sum_{k=1}^K \left [\sum_{(v_{n},v_{m}) \in E_k} R^{-1}_{n,m} \right]^{-1}
      \leq \er{uv},
      \mspace{32mu} \text{where  $(v_{n},v_{m})$ is an edge in the cutset $E_k$}.
    \end{equation}
  \end{lemma}
  \begin{figure}[H]
    \centering{
      \includegraphics[width=0.3\textwidth]{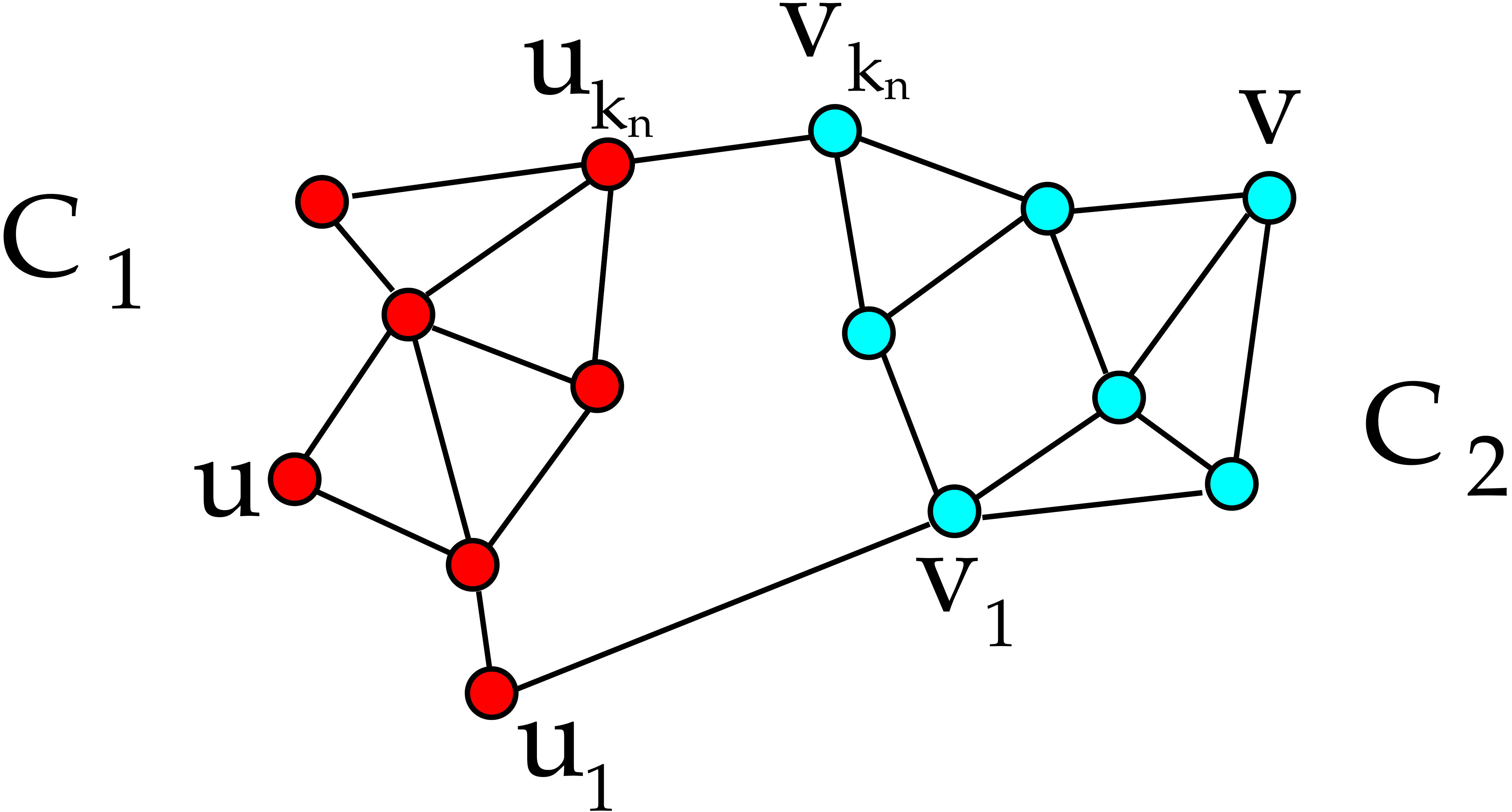}
    }
    \caption{Balanced, two   community stochastic blockmodel $G_n^{(k+1)}$ of size $n$ with $k_n$
      cross-community edges. The vertices $u$ and $v$ are in different communities, $u \in C_1$ and $v\in C_2$.
      \label{nash-williams-fig}}
  \end{figure}
  \noindent Since the set of cross-community edges is a cutset for all pairs of vertices $u$ and $v$
  in separate communities, and since the size of this set is precisely $k_n$, we immediately obtain
  the desired lower bound.

  The upper bound is obtained using the characterization of the effective resistance based on
  Thomson principle \cite{levin09}, which we recall briefly in the following. Let $f$ be a flow
  along the edges $E$ from $u$ to $v$, and let
  \begin{equation}
    {\cal E}(f) = \sum_{e\in E} f^2(e) R_e,
  \end{equation}
  be the energy of the flow $f$, where each undirected edge $e$ in the sum is only counted once. A unit flow has
  strength one,
  \begin{equation}
    \dv(f)(u) = - \dv (f)(v) = 1.
  \end{equation}
  Thomson's principle provides the following characterization of the effective resistance $\hR_{u\, w}$,
  \begin{equation}
    \er{uv}  = \min{} \left\{ {\cal E}(f), f \, \text{is a unit flow from $u$ to $v$} \right\}.
  \end{equation}
  We use Thomson's principle in the following way: we construct a unit flow $f$ from $u$ to $v$. For
  this flow, the energy ${\cal E}(f)$ yields an upper bound on $\hR_{u\, w}$.\\

  \noindent First, consider the case where neither $u$ nor $v$ are incident with any of the $k_n$
  cross-community edges, $e_i = (u_i, v_i), i = 1,\ldots,k_n$; where $u_i \in C_1$ and $v_i \in
  C_2$. Denote by $f^u_i$ the unit flow associated with the effective resistance between $u$ and
  $u_i$ when only the edges in $C_1$ are considered. Similarly define $f^v_i$ to be the unit flow
  associated with the effective resistance between $v$ and $v_i$ when when only the edges in $C_2$
  are considered. Using the corollary \ref{concentration}, given by (\ref{reff-dmean}),
  we have with high probability,
  \begin{equation}
    \mathcal{E}(f^u_i) = \frac{2}{\dbar} + \O{\frac{1}{\dbar^2}},
    \quad \text{and} \quad
    \mathcal{E}(f^v_i) = \frac{2}{\dbar} + \O{\frac{1}{\dbar^2}}.
  \end{equation}
  We note that the expression of $\mathcal{E}(f^v_i)$ should involve the expected degree in $C_2$.
  As explained in lemma \ref{n1_eq_n2}, we can use $\dbar$ since the difference between the two
  terms is absorbed in the $\O{\frac{1}{\dbar^2}}$ term. Finally, let $f^e_i$ be the flow that is 1
  on edge $e_i$ and 0 elsewhere. To conclude, we assemble the three flows and define
  \begin{equation*}
    f(e) = \frac{1}{k_n} \sum_{i=1}^k \left\{ f^u_i(e)+ f^e_i(e)+ f^v_i(e)\right\},
  \end{equation*}
  which is a unit flow from $u$ to $v$. Since $\mathcal{E}$ is a convex function, we can bound
  the energy of $f$ via
  \begin{align*}
    \mathcal{E}(f) 
    &= \mathcal{E}\left(\frac{1}{k_n} \sum_{i=1}^k \left\{ f^u_i(e)+ f^e_i(e)+ f^v_i(e)\right\}\right) 
      \leq 
      \mathcal{E}\left(\frac{1}{k_n} \sum_{i=1}^k f^u_i(e)\right) +
      \mathcal{E}\left(\frac{1}{k_n} \sum_{i=1}^k f^e_i(e)\right) + 
      \mathcal{E}\left(\frac{1}{k_n} \sum_{i=1}^k  f^v_i(e)\right) \\
    & \leq \frac{1}{k_n} \sum_{i=1}^k \mathcal{E}(f^u_i) + \frac{1}{k_n} + \frac{1}{k_n} \sum_{i=1}^k \mathcal{E}(f^v_i) 
      = \frac{4}{\dbar} + \frac{1}{k_n} + \O{\frac{1}{\dbar^2}}.
  \end{align*}
  The final line holds with high probability. Note that we calculate the energy of the flow in the
  center term directly, whereas convexity is used to estimate the energy in the first and third
  term.

  This upper bound also holds when either $u$ or $v$ is incident with any of the cross-community
  edges. In this case, $u=u_i$ for some $i$. For this $i$, we can formally define the flow $f^u_i$
  between $u$ and $u_i$ to be the zero flow, which minimizes the energy trivially and has energy
  equal to the resistance between $u$ and $u_i$ (which is zero). Then the above calculation yields a
  smaller upper bound for the first and third terms.
\end{proof}
\begin{remark}
  Lemma \ref{lemma:cross-comm} provides a first attempt at analysing the perturbation of the
  effective resistance under the addition of edges in the stochastic blockmodel. The upper bound
  provided by (\ref{lax-bound}) is too loose to be useful, and we therefore resort to a different
  technique to get a tighter bound. The idea is to observe that the effective resistance is
  controlled by the bottleneck formed by the cross-community edges. We can get very tight estimates
  of the fluctuations in the effective resistance using a detailed analysis of the addition of a
  single cross-community edge. We use the Sherman--Morrison--Woodburry theorem \cite{golan12} to
  compute a rank-one perturbation of the pseudo-inverse of the normalized graph Laplacian
  \cite{monnig16}, $\bL^\dagger$. The authors in \cite{ranjan14} provide us with the exact
  expression that is needed for our work, see \eqref{eq:update}. The proof proceeds by induction on
  the number of cross-community edges, $k_n$.
\end{remark}
\begin{Theorem}
  \label{thm:cross-comm}
  Let $G_{n}\sim\mathcal{G}(n,p_n,q_n)$ be a balanced, two community stochastic blockmodel with
  $p_n=\omega(\log n/n)$, $q_n=\omega(1/n^2)$, and $q_n = o(p_n/n)$. We assume that
  $G_n$ is connected. The effective resistance between two vertices $u$ and $v$ is given by
  \begin{equation}
    \hR_{u\, v} =
    \displaystyle  \frac{2}{\dbar} + 
    \begin{cases}
      \displaystyle  \O{\frac{1}{\dbar^2}}, & \text{if $u$ and  $v$ are in the same community,} \\
      \displaystyle \frac{1}{k_n} + \frac{\alpha(k_n,u,v)}{\dbar k_n} + \O{\frac{1}{\dbar^2}}, & \text{otherwise.} 
    \end{cases}
    \label{eq:theorem4}
  \end{equation}
  Also, conditioned on $k_n =k$ the random variable  $\alpha(k_n,u,v)$  is a deterministic function of
  $k$, and we have $\alpha(k_n,u,v) = \O{k_n}$.
\end{Theorem}
\begin{proof}
  First, observe that Lemma \ref{lemma:cross-comm} immediately implies that
  \begin{equation}
    -2k_n\leq \alpha(k_n,u,v) \leq 2k_n,\quad \forall u,v,
    \label{eq:bounded-k}
  \end{equation}
  with high probability, so $\alpha=\mathcal{O}(k_n)$.

  Next, let us show that the in-community resistances follow the prescribed form. The proof proceeds
  as follows: we derive the expression \eqref{eq:theorem4} conditioned on the random variable
  $k_n =k$, and we prove that $\alpha (k,u,v)$ is indeed a deterministic function in this case; the
  derivation of \eqref{eq:theorem4} is obtained by induction on $k$.

  The engine of our induction is the update formula (equation 11) in \cite{ranjan14}. This provides an
  exact formula (equation \eqref{eq:update} below) for the change in resistance between any pair of
  vertices in a graph when a single edge is added or removed. The particular motivation of the
  authors in \cite{ranjan14} is to calculate rank one updates to the pseudoinverse of the
  combinatorial graph Laplacian; however, it conspires that their formula is also very useful to
  inductively calculate resistances in the stochastic blockmodel.

  We first consider the base case, where $G_n$ is a balanced, two community stochastic blockmodel of
  size $n$ with $k_n =1 $ cross-community edge. Denote this edge by $e_1=(u_1,v_1)$, where $u_1\in C_1$
  and $v_1\in C_2$. We will refer to this graph as $G_n^{(1)}$.
  
  The addition of a single edge connecting otherwise disconnected components does not change the
  resistance within those components, as it does not introduce any new paths between two vertices
  within the same component. Because each community is an \ER graph with parameters $p_n$ and
  respective sizes $n_1 = \lfloor (n+1)/2 \rfloor$ and $n_2 = \lfloor n/2 \rfloor$, corollary
  \ref{concentration} provides the expression for the effective resistance between two vertices
  within each community.  A simple circuit argument allows us to obtain the resistance between $u$
  and $v$ in separate communities via
  \begin{equation}
    \label{eq:1001}
    \hR_{u\, v} = \hR_{u\, u_1} + \hR_{u_1,v_1} + \hR_{v_1,v}.
  \end{equation}
  If $u\neq u_1$ and $v \neq v_1$, then we combine Nash-Williams and corollary
  \ref{concentration} to get
  \begin{equation*}
    \hR_{u\, v} = \frac{1}{k}+\frac{4}{\dbar} + \mathcal{O}\left(\frac{1}{\dbar^2}\right).
  \end{equation*}
  If $u=u_1$ and/or $v=v_1$ then the appropriate resistances are  set to zero in
  \eqref{eq:1001}. \\

  \noindent In summary, for arbitrary pairs $(u,v)$ in $G^{(1)}_n$, we have
  \begin{equation*}
    \hR_{u\,v} = \frac{1}{k}+\frac{2}{\dbar}+\frac{\alpha(u,v)}{k\dbar} +
    \mathcal{O}\left(\frac{1}{\dbar^2}\right), 
  \end{equation*}
  where 
  \begin{equation}
    \alpha(u,v) =
    \begin{cases}
      0 & \text{if $u=u_1$ and $v=v_1$},\\
      1 & \text{if $u=u_1$ and $v\neq v_1$, or $u \neq u_1$ and $v=v_1$},\\
      2 & \text{if $u\neq u_1$ and $v\neq v_1$.}
    \end{cases}
  \end{equation}
  This establishes the base case for \eqref{eq:theorem4}.\\

  We now assume that \eqref{eq:theorem4} holds for any balanced, two community stochastic
  blockmodel of size $n$ with $k_n = k$ cross-community edges.  We consider a balanced, two
  community stochastic blockmodel $G_n^{(k+1)}$ of size $n$ with $k_n = k+1$ cross-community
  edges. We denote the cross-community edges by $e_i=(u_i,v_i), i = 1,\ldots, k+1$, where
  $u_i\in C_1$ and $v_i\in C_2$ (see Fig. \ref{same-community}). 

  Finally, we denote by $G_n^{(k)}$ the balanced, two  community stochastic blockmodel with $k$ 
  cross-community edges obtained by removing the edge $e_{k+1}=(u_{k+1},v_{k+1})$ from
  $G_n^{(k+1)}$.

  Let $\hR$ denote the effective resistances in $G_n^{(k)}$ and $\hR'$ denote the
  effective resistances in $G_n^{(k+1)}$. 

  Since $G_n^{(k)}$ is obtained by removing an edge
  from $G_n^{(k+1)}$, we can apply equation (11) in \cite{ranjan14} to express $\hR'$ from
  $\hR$,
  
  \begin{equation}
    \hR'_{u\, v} = \hR_{u\, v} - 
    \frac{
      \left[
        (\hR_{u\, v_{k+1}}-\hR_{u\, u_{k+1}}) - (\hR_{v\, v_{k+1}} - \hR_{v\, u_{k+1}})
      \right]^2}
    {4(1+\hR_{u_{k+1}\,v_{k+1}})}. 
    \label{eq:update}
  \end{equation}
  \begin{figure}[H]
    \centering{
      \includegraphics[width=0.3\textwidth]{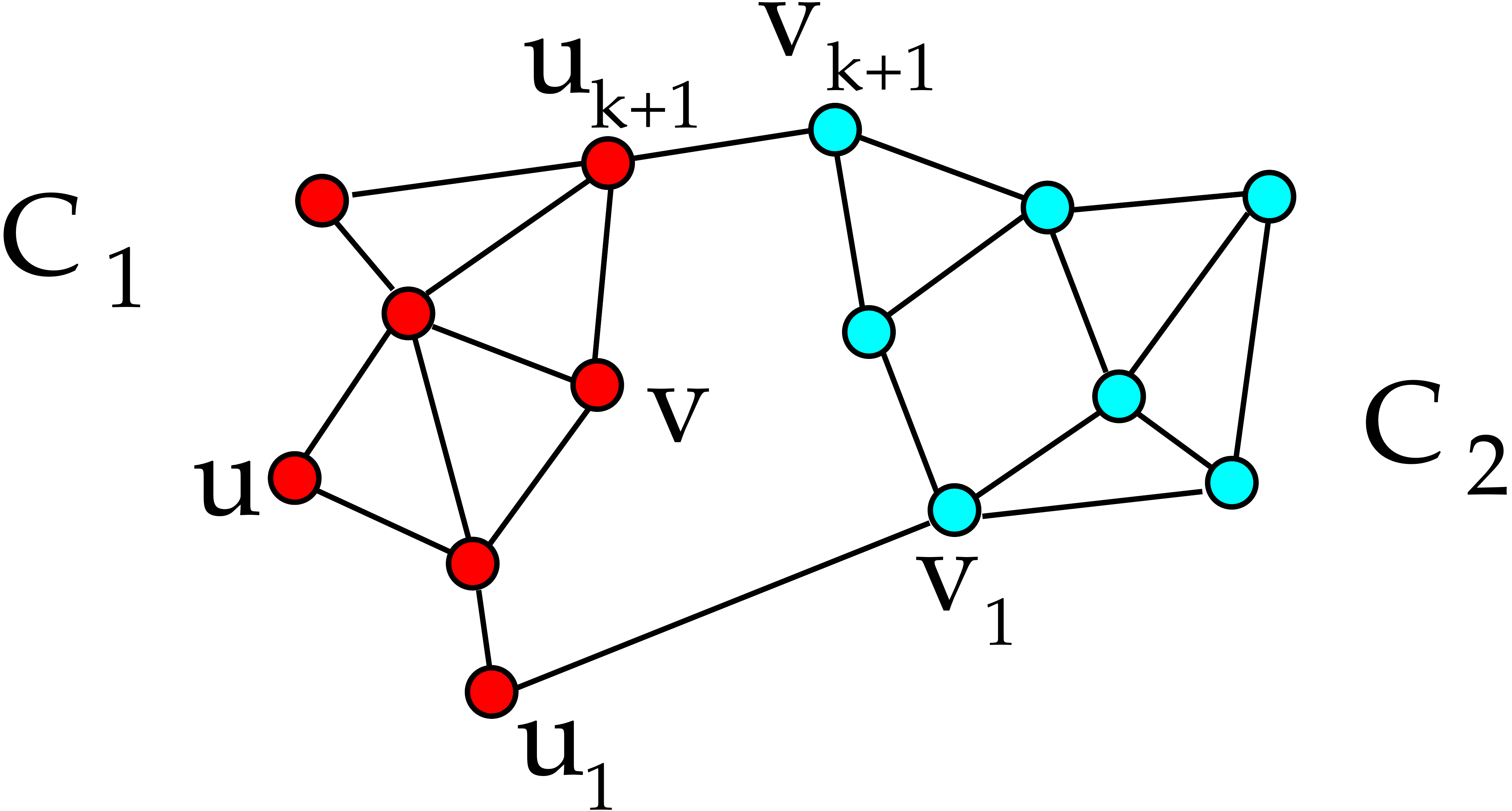}
    }
    \caption{Balanced, two   community stochastic blockmodel $G_n^{(k+1)}$ of size $n$ with $k_n = k+1$
      cross-community edges. Vertices $u$ and $v$ are in the same community $C_1$
      \label{same-community}}
  \end{figure}
  \noindent In the following we use the induction hypothesis to compute $\hR'$ using
  \eqref{eq:update}. We first consider the case where the vertices $u$ and $v$ belong to the same
  community, say $C_1$ without loss of generality (see Fig. \ref{same-community}).

  We need to consider the following three possible scenarios: 
  \begin{enumerate}
  \item $u \neq u_{k+1}$ and $v \neq v_{k+1}$,
  \item $u = u_{k+1}$ and $v \neq v_{k+1}$,
  \item $u \neq u_{k+1}$ and $v = v_{k+1}$.
  \end{enumerate}
  We will treat the first case; the last two cases are in fact equivalent, and are straightforward
  consequences of the analysis done in the first case. From the induction hypothesis we have
  \begin{align}
    \hR_{u\, v_{k+1}} & = \frac{1}{k} + \frac{2}{\dbar} + \frac{\alpha(k, u,v_{k+1})}{k\dbar} +\O{1/\dbar^2},\\
    \hR_{u\, u_{k+1}} & =  \frac{2}{\dbar} + \O{1/\dbar^2},\\
    \hR_{v\, v_{k+1}}& =  \frac{1}{k} + \frac{2}{\dbar} + \frac{\alpha(k, v,v_{k+1})}{k\dbar} +\O{1/\dbar^2},\\
    \hR_{v\, u_{k+1}} & =  \frac{2}{\dbar} + \O{1/\dbar^2},\\
    \hR_{u_{k+1}\,v_{k+1}} & = \frac{1}{k} + \frac{2}{\dbar} + \frac{\alpha(k, u_{k+1}\,v_{k+1})}{k\dbar} +\O{1/\dbar^2}.
  \end{align}
  Substituting these expression into \eqref{eq:update}, we get
  \begin{align}
    \hR'_{u\, v} & = \hR_{u\, v} - \frac{
                 \displaystyle  \left[ 
                 \left(\frac{1}{k} + \frac{\alpha(k, u,v_{k+1})}{k\dbar}\right) - 
                 \left(\frac{1}{k} + \frac{\alpha(k, v,v_{k+1})}{k\dbar}\right) +\O{1/\dbar^2}
                 \right]^2}
                 {\displaystyle 4(1+ \frac{1}{k} + \frac{2}{\dbar} + \frac{\alpha(k, u_{k+1}\,v_{k+1})}{k\dbar} +\O{1/\dbar^2})}\\
               & = \hR_{u\, v} -  \O{\frac{1}{\dbar^2}} \frac{
                 \left[
                 \frac{\displaystyle  \alpha(k, u,v_{k+1})- \alpha(k, v,v_{k+1})}{k} +\O{1/\dbar}
                 \right]^2}
                 {\displaystyle 4\left(1+ \frac{1}{k} + \frac{2}{\dbar} + \frac{\alpha(k, u_{k+1}\,v_{k+1})}{k\dbar}
                 +\O{1/\dbar^2}\right)} 
  \end{align}
  Because $\alpha$ is bounded with high probability (see \eqref{eq:bounded-k}), we have
  \begin{equation}
    \frac{  \displaystyle  \alpha(k, u,v_{k+1})- \alpha(k, v,v_{k+1})}{k} +\O{1/\dbar} = \O{1} \quad \text{with high probability},
  \end{equation}
  and also
  \begin{equation}
    1+ \frac{1}{k} + \frac{2}{\dbar} + \frac{\alpha(k, u_{k+1}\,v_{k+1})}{k\dbar}  +\O{1/\dbar^2}
    = 1+ \frac{1}{k} +\O{1/\dbar} \quad \text{with high probability},
  \end{equation}
  which implies
  \begin{equation}
    \frac{\left[
        \frac{\displaystyle  \alpha(k, u,v_{k+1})- \alpha(k, v,v_{k+1})}{k} +\O{1/\dbar}
      \right]^2}
    {\displaystyle 4\left(1+ \frac{1}{k} + \frac{2}{\dbar} + \frac{\alpha(k, u_{k+1}\,v_{k+1})}{k\dbar}
        +\O{1/\dbar^2}\right)} 
    = \O{1} \quad \text{with high probability}.
  \end{equation}
  \begin{figure}[H]
    \centering{
      \includegraphics[width=0.3\textwidth]{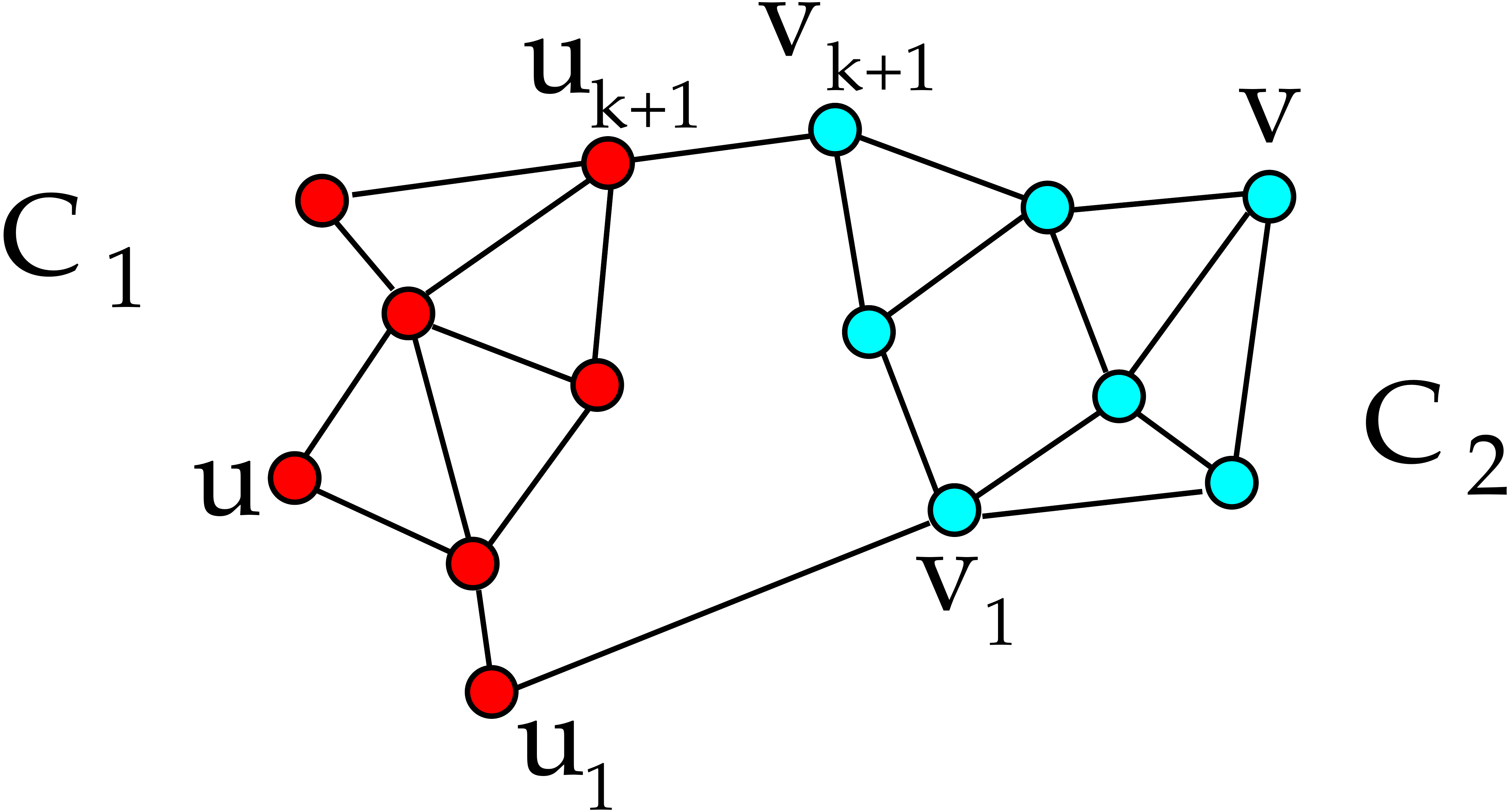}
    }
    \caption{Balanced, two   community stochastic blockmodel $G_n^{(k+1)}$ of size $n$ with $k_n = k+1$
      cross-community edges. The vertices $u$ and $v$ are in different communities, $u \in C_1$ and $v\in C_2$.
      \label{different-communities}}
  \end{figure}
  \noindent   We conclude that 
  \begin{equation}
    \hR'_{u\, v}  = \hR_{u\, v}  + \O{1/\dbar^2} \quad \text{with high probability}.
  \end{equation}
  This completes the induction, and the proof of \eqref{eq:theorem4} in the case where $u$ and $v$
  belong to the same community.\\
  
  We now consider the case where $u \in C_1$ and $v \in C_2$ (see
  Fig. \ref{different-communities}). As above, we need to consider the following three possible
  scenarios:
  \begin{enumerate}
  \item $u \neq u_{k+1}$ and $v \neq v_{k+1}$,
  \item $u = u_{k+1}$ and $v \neq v_{k+1}$,
  \item $u \neq u_{k+1}$ and $v = v_{k+1}$.
  \end{enumerate}
  Again,  we only  prove the first case; the last two equivalent cases are straightforward
  consequences of the first case. From the induction hypothesis we now have
  \begin{align}
    \hR_{u\, u_{k+1}} & =  \frac{2}{\dbar} + \O{1/\dbar^2},\\
    \hR_{v\, v_{k+1}}& =  \frac{2}{\dbar} + \O{1/\dbar^2},\\
    \hR_{u\, v_{k+1}} & = \frac{1}{k} + \frac{2}{\dbar} + \frac{\alpha(k, u,v_{k+1})}{k\dbar} +\O{1/\dbar^2},\\
    \hR_{v\, u_{k+1}} & =  \frac{1}{k} + \frac{2}{\dbar} + \frac{\alpha(k, v,v_{k+1})}{k\dbar} +\O{1/\dbar^2},\\
    \hR_{u_{k+1}\, v_{k+1}} & = \frac{1}{k} + \frac{2}{\dbar} + \frac{\alpha(k, u_{k+1}\, v_{k+1})}{k\dbar} +\O{1/\dbar^2}.
  \end{align}
  To reduce notational clutter, we use some abbreviated notation to denote the various $\alpha$
  terms associated with the vertices of interest (see Fig. \ref{thealphas}),
  \begin{align*}
    \alpha_0 &\eqdef \alpha(k_n^{(k)},u,v), \\
    \alpha_1 &\eqdef \alpha(k_n^{(k)},u_{k+1}\, v), \\
    \alpha_2 &\eqdef \alpha(k_n^{(k)},u,v_{k+1}), \\
    \alpha_3 &\eqdef \alpha(k_n^{(k)},u_{k+1}\, v_{k+1}).
  \end{align*}
  \begin{figure}[H]
    \centering{
      \includegraphics[width=0.25\textwidth]{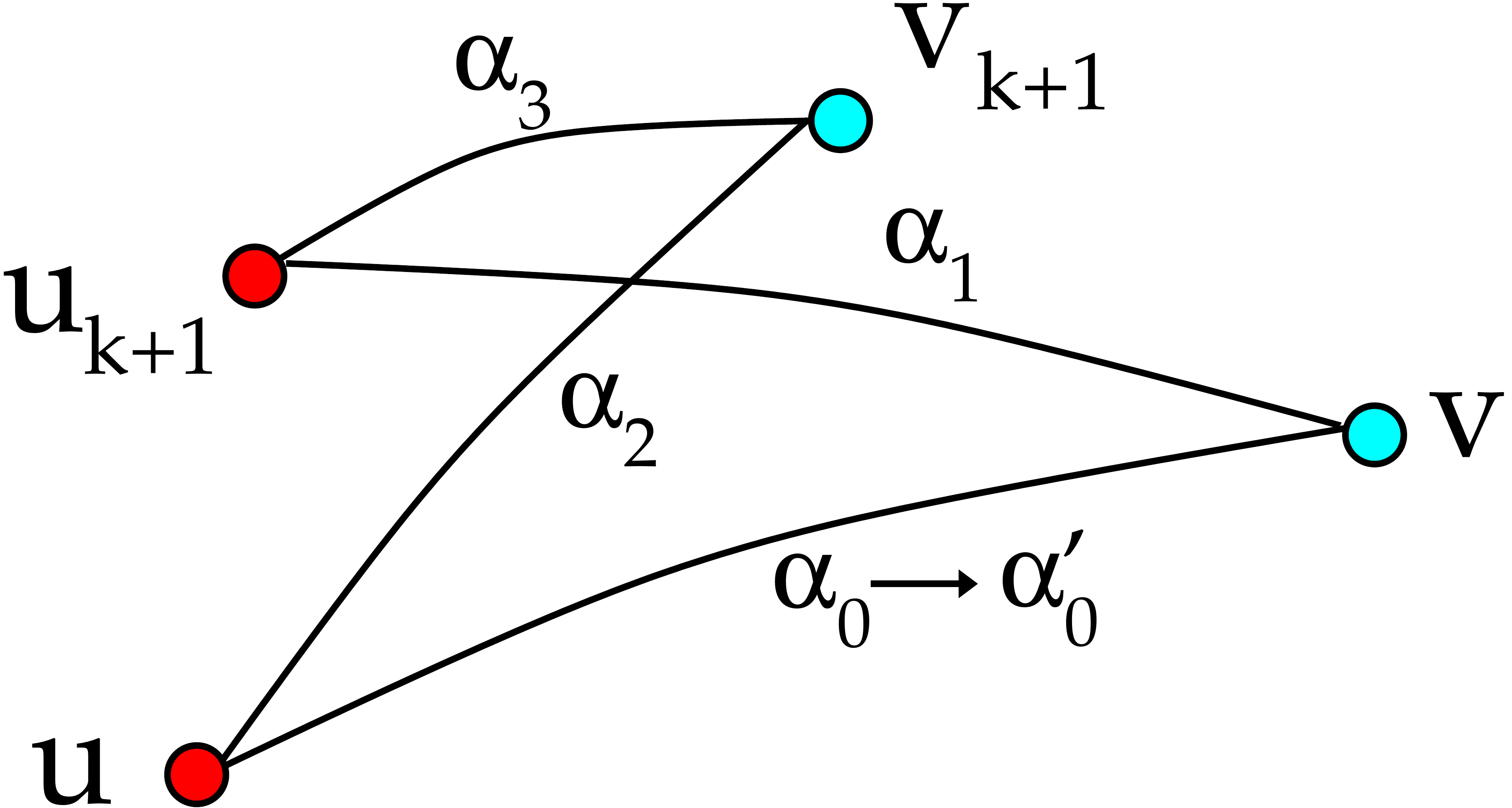}
    }
    \caption{Coefficients $\alpha_0,\ldots,\alpha_3$ in \eqref{eq:theorem4} for several pairs of vertices
      in $G_n^{(k+1)}$.
      \label{thealphas}}
  \end{figure}
  Let us denote the decrease in effective resistance by $\Delta \hR$,
  \begin{equation}
    \DhR \eqdef \hR_{u\, v} - \hR'_{u\, v}.
  \end{equation}
  From \eqref{eq:update}, we have
  \begin{equation}
    \begin{split}
      \DhR & = \frac{
        \displaystyle  \left[ 
          \left(\frac{1}{k} + \frac{\alpha_2}{k\dbar}\right) - 
          \left(-\frac{1}{k} - \frac{\alpha_1}{k\dbar}\right) +\O{1/\dbar^2}
        \right]^2}
      {4\left(\displaystyle 1+ \frac{1}{k} + \frac{2}{\dbar} + \frac{\alpha_3}{k\dbar} +\O{1/\dbar^2}\right)}
      =  \frac{
        \left[\displaystyle
          \frac{2}{k}\left( 1 + 
            \frac{\displaystyle  \alpha_1+ \alpha_2}{2\dbar} \right) +\O{1/\dbar^2}
        \right]^2}
      {4\left(\displaystyle 1+ \frac{1}{k} + \frac{2}{\dbar} + \frac{\alpha_3}{k\dbar} +\O{1/\dbar^2}\right)}\\
      & =  \frac{1}{\displaystyle k^2}
      \left[  1 +  \frac{\displaystyle  \alpha_1+ \alpha_2}{2\dbar} +\O{1/\dbar^2}
      \right]^2
      \frac{1}
      {\displaystyle 1+ \frac{1}{k} + \frac{2}{\dbar} + \frac{\alpha_3}{k\dbar} +\O{1/\dbar^2}}\\
      & =  \left(
        \frac{1}{\displaystyle k^2}   +  
        \frac{\displaystyle  \alpha_1 + \alpha_2}{\dbar k^2} +
        \frac{\displaystyle  (\alpha_1 + \alpha_2)^2}{4(\dbar k)^2} + \O{1/\dbar^2}
      \right)
      \frac{1}
      {\displaystyle 1+ \frac{1}{k} + \frac{2}{\dbar} + \frac{\alpha_3}{k\dbar} +\O{1/\dbar^2}}\\
      & =  \left(
        \frac{1}{\displaystyle k^2}   +  
        \frac{\displaystyle  \alpha_1 + \alpha_2}{\dbar k^2} + \O{1/\dbar^2}
      \right)
      \frac{1}
      {\displaystyle 1+ \frac{1}{k} + \frac{2}{\dbar} + \frac{\alpha_3}{k\dbar} +\O{1/\dbar^2}}.
    \end{split}
    \label{eq:dr1}
  \end{equation}
  
  At this juncture, we need to expand $\left(1+ 1/k + 2/\dbar +
    \alpha_3/(k\dbar) +\O{1/\dbar^2}\right)^{-1}$ using a Taylor series, which is possible when 
  $1/k + 2/\dbar +  \alpha_3/(k\dbar) < 1$. Since we are interested in  the large $n$ asymptotic, we
  can assume that for $n$ sufficiently large $2/\dbar +  \alpha_3/(k\dbar) < 1/2$, and thus we need to
  guarantee that $k\ge 2$. \\

  \noindent The case $k=1$ needs to be handled separately. Setting $k=1$ into \eqref{eq:dr1} yields
  \begin{equation}
    \DhR  =  \frac{1+ \frac{\displaystyle  \alpha_1 + \alpha_2}{\dbar} + \O{1/\dbar^2}}
    {\displaystyle 2 + \frac{2 + \alpha_3}{\dbar} +\O{1/\dbar^2}}
    = \frac{1}{2}
    \frac{1+ \frac{\displaystyle \alpha_1 + \alpha_2}{\dbar} + \O{1/\dbar^2}}
    {\displaystyle 1 + \frac{2 + \alpha_3}{2\dbar} +\O{1/\dbar^2}},
  \end{equation}
  for $n$ sufficiently large $(2 +  \alpha_3)/\dbar < 1$, thus 
  \begin{equation}
    \begin{split}
      \DhR         & = \frac{1}{2}
      \left(1+ \frac{\displaystyle \alpha_1 + \alpha_2}{\dbar} + \O{1/\dbar^2} \right)
      \left(1 - \frac{2 + \alpha_3}{2\dbar} + \O{1/\dbar^2}\right)
      = \frac{1}{2}
      \left(1 + \frac{\displaystyle 2 \alpha_1 + 2\alpha_2 - 2 - \alpha_3}{2\dbar} +
        \O{1/\dbar^2}\right) \\
      & = \frac{1}{2} + \frac{\displaystyle \alpha_1 + \alpha_2 - 1 - \alpha_3/2}{2\dbar} +
      \O{1/\dbar^2},
    \end{split}
  \end{equation}
  which leads to 
  \begin{equation}
    \begin{split}
      \hR'_{u\, v} & = \hR_{u\, v} - \DhR 
      = 1 + \frac{2}{\dbar} + \frac{\alpha_0}{\dbar} -\frac{1}{2} - \frac{\displaystyle \alpha_1 +
        \alpha_2 - 1 - \alpha_3/2}{2\dbar} +  \O{1/\dbar^2},\\
      & = \frac{1}{2} + \frac{2}{\dbar} + \frac{\displaystyle 2\alpha_0 - \alpha_1 - \alpha_2 +
        1 + \alpha_3/2}{2\dbar} + \O{1/\dbar^2}.
    \end{split}
  \end{equation}
  Let  
  \begin{equation}
    \alpha'_0 \eqdef 2\alpha_0 - \alpha_1 - \alpha_2 + 1 + \alpha_3/2,
    \label{eq:update-alpha1}
  \end{equation}
  then we have
  \begin{equation}
    \hR'_{u\, v}  = \frac{1}{2} + \frac{2}{\dbar} + \frac{\displaystyle \alpha'_0}{2\dbar} + \O{1/\dbar^2},
  \end{equation}
  which matches the expression given in \eqref{eq:theorem4} for $k=2$. This completes the induction for
  $k=1$.\\

  \noindent We now proceed to the general case where $k\ge 2$. In that case, we use a Taylor series
  expansion of $\left(1+ 1/k + 2/\dbar +  \alpha_3/(k\dbar) +\O{1/\dbar^2}\right)^{-1}$, and we get
  \begin{equation}
    \left(1+ 1/k + 2/\dbar +  \alpha_3/(k\dbar) +\O{1/\dbar^2}\right)^{-1}  = \sum_{m=0}^\infty
    (-1)^m \left(\frac{1}{k} + \frac{2}{\dbar} + \frac{\alpha_3}{k\dbar}\right)^m +\O{1/\dbar^2}.
    \label{eq:taylor1}
  \end{equation}
  Now, most of the term in $\left(\frac{1}{k} + \frac{2}{\dbar} + \frac{\alpha_3}{k\dbar}\right)^m $
  are of order $\O{1/\dbar^2}$, and we need to carefully extract the few significant terms. In the
  expansion of $\left(\frac{1}{k} + \frac{2}{\dbar} + \frac{\alpha_3}{k\dbar}\right)^m $ the only
  terms that do not contain a $1/\dbar^2$ are obtained by choosing systematically $1/k$ in each
  of the $m$ factors, or choosing $1/k$ in all but one factors and either $\frac{2}{\dbar}$ or
  $\frac{\alpha_3}{k\dbar}$ in the last factor. There are $m$ ways to construct these last two
  terms. In summary, we have  for $m \ge 1$,
  \begin{equation}
    \left(\frac{1}{k} + \frac{2}{\dbar} + \frac{\alpha_3}{k\dbar}\right)^m  = \frac{1}{k^m} + m
    \frac{\alpha_3}{k^{m-1}\dbar} + m\frac{2}{k^{m-1}\dbar} + \O{1/\dbar^2}.
    \label{eq:multinomial}
  \end{equation}
  We can substitute \eqref{eq:multinomial}into \eqref{eq:taylor1} to get
  \begin{equation}
    \begin{split}
      \left(1+ 1/k + 2/\dbar +  \alpha_3/(k\dbar) +\O{1/\dbar^2}\right)^{-1}  & = \sum_{m=0}^\infty
      (-1)^m \frac{1}{k^m} + \sum_{m=1}^\infty m \frac{\alpha_3}{k^{m-1}\dbar} + m\frac{2}{k^{m-1}\dbar} + \O{1/\dbar^2}\\
      & = \frac{1}{1+ 1/k} - \frac{\alpha_3}{\dbar} \frac{k}{(k+1)^2} - 2 \frac{k^2}{(k+1)^2} +
      \O{1/\dbar^2}\\
      & = \frac{k}{k + 1} - \frac{(\alpha_3 + 2k)k}{\dbar (k+1)^2} + \O{1/\dbar^2}.
      \label{eq:taylor2}
    \end{split}
  \end{equation}
  We can insert \eqref{eq:taylor2} into \eqref{eq:dr1} to get
  \begin{equation}
    \begin{split}
      \DhR & = \left(
        \frac{1}{\displaystyle k^2}   +  
        \frac{\displaystyle  \alpha_1 + \alpha_2}{\dbar k^2} + \O{1/\dbar^2}
      \right)
      \left(\frac{k}{k + 1} - \frac{(\alpha_3 + 2k)k}{\dbar (k+1)^2} +
        \O{1/\dbar^2}\right)\\
      & = \frac{1}{k(k+1)} + \frac{\alpha_1 + \alpha_2}{\dbar k (k+1)} 
      - \frac{(\alpha_3 + 2k)}{\dbar k(k+1)^2} + \O{1/\dbar^2},
    \end{split}
  \end{equation}
  which leads to 
  \begin{equation}
    \begin{split}
      \hR'_{u\, v} & = \hR_{u\, v} - \DhR \\
      & = \frac{1}{k}  + \frac{2}{\dbar} + \frac{\alpha_0}{\dbar k} - \frac{1}{k(k+1)} - \frac{\alpha_1 + \alpha_2}{\dbar k (k+1)} 
      + \frac{(\alpha_3 + 2k)}{\dbar k(k+1)^2} +  \O{1/\dbar^2}\\
      & = \frac{1}{k+1} + \frac{2}{\dbar} + \left\{\alpha_0\frac{k+1}{k} - \frac{\alpha_1 +\alpha_2}{k}
        + \frac{\alpha_3 + 2k}{ k(k+1)}\right\}\frac{1}{\dbar(k+1)} + \O{1/\dbar^2}.
    \end{split}
  \end{equation}
  Let  
  \begin{equation}
    \alpha'_0 \eqdef \alpha_0 + \frac{\alpha_0 -\alpha_1 -\alpha_2}{k} + \frac{2}{k+1}
    + \frac{\alpha_3}{k(k+1)},
    \label{eq:update-alpha}
  \end{equation}
  then we have
  \begin{equation}
    \hR'_{u\, v}  = \frac{1}{k+1} + \frac{2}{\dbar} + \frac{\displaystyle \alpha'_0}{\dbar(k+1)} + \O{1/\dbar^2},
  \end{equation}
  which matches the expression given in \eqref{eq:theorem4} for $k+1$. This completes the induction
  and the proof of \eqref{eq:theorem4}.

  \noindent We note that $\alpha'_0 = \alpha(k+1,u,v)$ given by \eqref{eq:update-alpha} is a deterministic
  function of $k$, since  $\alpha_0, \alpha_2, \alpha_2$ and $\alpha_3$ all are deterministic
  functions of $k$, by the inductive hypothesis.

  \noindent The cases where $u =  u_{k+1}$ or $v =  v_{k+1}$ are treated similarly; since they present
  no new difficulties they are omitted.

  This completes the proof of the  theorem.
\end{proof}
\begin{remark}
  As expected \eqref{eq:update-alpha} agrees with the update formula in the case $k=1$, given by \eqref{eq:update-alpha1}.
\end{remark} 
\noindent The update formula \eqref{eq:update-alpha} implies that the {\em random variables} $\alpha(k_n,u,v)$
and $\alpha(k_n,u,v)/k_n,$ both have bounded variation, as explained in the following corollary.
\begin{corollary}
  Let $G_{n+1}\sim\cG(n+1,p_n,q_n)$ be a stochastic blockmodel with $p_n=\om{\log n/n}$. Assume that
  $q_n = \o{p_n/n}$. Let $G_{n}$ be the subgraph of $G_{n+1}$ induced by the vertex set
  $[n]$. Let $k_n$ and $k_{n+1}$ be the number of cross-community edges in $G_{n}$ and  $G_{n+1}$
  respectively. \\

  Let $u$ and $v$ be two vertices, and let  $\hR^{(n)}_{u\, v}$ and  $\hR^{(n+1)}_{u\, v}$ be the
  effective resistances measured in $G_n$ and $G_{n+1}$ respectively.

  Let $\alpha(k_n,u,v)$ and $\alpha(k_{n+1},u,v)$ be the coefficients  in the expansion of
  $\hR^{(n)}_{u\, v}$ and $\hR^{(n+1)}_{u\, v}$ in \eqref{eq:theorem4} respectively. We have,
  \begin{align}
    \left|\alpha(k_{n+1},u,v) - \alpha(k_n,u,v)\right| & = \O{1} \quad \text{with high probability,}\\
    \\
    \displaystyle 
    \left|\frac{\alpha(k_{n+1},u,v)}{k_{n+1}} - \frac{\alpha(k_n,u,v)}{k_n} \right| & = \O{\frac{1}{k_n}} \quad \text{with high probability.}\\
  \end{align}
  \label{bounded-alpha-kn}
\end{corollary}
\noindent The proof of the corollary is a consequence of the following proposition which shows that,
conditioned on $k_n = k$, the {\em functions} $\alpha(k,u,v)$ and $\alpha(k,u,v)/k$, defined for
$k \ge 1$, have bounded variation.
\begin{proposition}
  Let $G_{n}\sim\mathcal{G}(n,p_n,q_n)$ be a balanced, two community stochastic blockmodel with
  $p_n=\omega(\log n/n)$, $q_n=\omega(1/n^2)$, and $q_n = o(p_n/n)$. We assume that $G_n$ is
  connected. Let $u$ and $v$ be two vertices. Given $k_n=k$, let $\alpha(k,u,v)$ be the coefficient
  in the expansion of $\hR_{u\, v}$ in \eqref{eq:theorem4}. Similarly, let $\alpha(k+1,u,v)$ be the
  corresponding quantity when $k_n =k+1$. We have
  \begin{align}
    \left|\alpha(k+1,u,v) - \alpha(k,u,v)\right| & \leq 8, \\
    \\
    \displaystyle 
    \left|\frac{\alpha(k+1,u,v)}{k+1} - \frac{\alpha(k,u,v)}{k} \right| & \leq  \frac{6}{k}.
  \end{align}
  \label{bounded-alpha-k}
\end{proposition}
\begin{proof}
  The proof is a direct consequence of \eqref{eq:update-alpha}, and the fact that $|\alpha(k,u,v)|
  \leq 2k$ (see \eqref{eq:bounded-k}).
  \noindent Let us start with the first inequality.  From \eqref{eq:update-alpha} we have
  \begin{equation}
    \alpha(k+1,u,v) - \alpha(k,u,v) =  \frac{\alpha(k,u,v) -\alpha_1 -\alpha_2}{k} + \frac{2}{k+1}     + \frac{\alpha_3}{k(k+1)},
  \end{equation}
  and thus
  \begin{equation}
    \left|\alpha(k+1,u,v) - \alpha(k,u,v) \right|  \leq \frac{6k}{k} +
    \frac{2}{k+1}  + \frac{2k}{k(k+1)} =  6 + \frac{4}{k+1} \leq 8.
  \end{equation}
  We now  show the second inequality. Again, from \eqref{eq:update-alpha} we have
  \begin{equation}
    \begin{split}
      \frac{\alpha(k+1,u,v)}{k+1} - \frac{\alpha(k,u,v)}{k} & = \frac{\alpha(k,u,v)}{k+1} - \frac{\alpha(k,u,v)}{k}
      + \frac{\alpha(k,u,v) -\alpha_1 -\alpha_2}{k(k+1)} + \frac{2}{(k+1)^2}     + \frac{\alpha_3}{k(k+1)^2}\\
      & = - \frac{\alpha(k,u,v)}{k(k+1)} + \frac{\alpha(k,u,v) -\alpha_1 -\alpha_2}{k(k+1)} + \frac{2}{(k+1)^2}  + \frac{\alpha_3}{k(k+1)^2}\\
      & = -\frac{ \alpha_1 + \alpha_2}{k(k+1)} + \frac{2}{(k+1)^2}  + \frac{\alpha_3}{k(k+1)^2},\\
    \end{split}
  \end{equation}
  and thus
  \begin{equation}
    \begin{split}
      \left|\frac{\alpha(k+1,u,v)}{k+1} - \frac{\alpha(k,u,v)}{k} \right| & \leq \frac{4k}{k(k+1)} +
      \frac{2}{(k+1)^2}  + \frac{2k}{k(k+1)^2} =  \frac{4}{k+1} + \frac{2}{(k+1)^2}  + \frac{2}{(k+1)^2}\\
      & \leq \frac{4}{k+1} \left\{1 +  \frac{1}{k+1}\right\} \leq \frac{6}{k+1} \leq \frac{6}{k}. 
    \end{split}
  \end{equation}
\end{proof}
We now proceed to the proof of corollary \ref{bounded-alpha-kn}.
\begin{proof}[Proof of corollary \ref{bounded-alpha-kn}]
  We first verify that with high probability $k_{n+1} - k_n$ is bounded. We then apply
  proposition~\ref{bounded-alpha-k}.   

  Let $\Delta k_n \eqdef k_{n+1} - k_n$ be the number of adjacent cross-community edges that have
  vertex $n+1$ as one of their endpoints. $\Delta k_n$ is a binomial random variable $B(n_1,q_n)$,
  where we assume without loss of generality that vertex $n+1 \in C_2$. Because
  $\E{\Delta k_n} = nq_n/2 = p_n \o{1}$, $\E{\Delta k_n}$ is bounded, and there exists $\kappa$ such
  that $\forall n, \E{\Delta k_n} <\kappa$. Using a Chernoff bound we have
  \begin{equation}
    \prob{|\E{\Delta k_n} - \Delta k_n| > 3 \sqrt{\kappa}} \leq 2 \exp{\left(-\frac{1}{3}\right)} <0.01,
  \end{equation}  
  and thus 
  \begin{equation}
    \left|k_{n+1} - k_n\right| = \left|k_{n+1} - \E{\Delta k_n} + \E{\Delta k_n}  - k_n\right| < 6
    \sqrt{\kappa} \quad \text{with probability > 0.99}
  \end{equation}  
  \noindent In other words, with high probability $k_{n+1} - k_n$ is bounded by $C = \lceil
  6\sqrt{\kappa} \rceil$, independently of $n$.\\ 
  Finally, we have
  \begin{equation}
    \begin{split}
      \left|\frac{\alpha(k_{n+1},u,v)}{k_{n+1}} - \frac{\alpha(k_n,u,v)}{k_n}\right|
      & = \left|\sum_{k = k_n}^{k_{n+1}-1} \frac{\alpha(k+1,u,v)}{k+1} - \frac{\alpha(k,u,v)}{k}\right|\\
      & \leq \sum_{k = k_n}^{k_{n+1}-1} \left|\frac{\alpha(k+1,u,v)}{k+1} -
        \frac{\alpha(k,u,v)}{k}\right|\\
    \end{split}
  \end{equation}
  Because of corollary~\ref{bounded-alpha-k}, each term in the sum is bounded by $6/k$; the largest
  upper bound being $6/k_n$. Also, with high probability there are at most $C$ terms. We conclude
  that 
  \begin{equation}
    \left|\frac{\alpha(k_{n+1},u,v)}{k_{n+1}} - \frac{\alpha(k_n,u,v)}{k_n}\right| \leq C
    \frac{6}{k_n}\quad \text{with high probability.}
  \end{equation}
\end{proof}

We now combine lemma \ref{n1_eq_n2} with theorem \ref{thm:cross-comm} to estimate the perturbation
created by the addition of an $n+1^\text{th}$ vertex to $G_{n}$.  The following lemma shows that
adding an additional vertex, with corresponding edges to either one of the communities does not
change the effective resistance, as long as no new cross-community edges are created. 
\begin{lemma}
  Let $G_{n+1}\sim\cG(n+1,p_n,q_n)$ be a stochastic blockmodel with $p_n=\om{\log n/n}$. Assume that
  $q_n = \o{p_n/n}$. Let $G_{n}$ be the subgraph of $G_{n+1}$ induced by the vertex set
  $[n]$. Let $k_n$ and $k_{n+1}$ be the number of cross-community edges in $G_{n}$ and  $G_{n+1}$
  respectively. \\

  \noindent Let $u$ and $v$ be two vertices in $G_{n+1}$, for which the effective resistance
  $\hR^{(n)}_{u\, v}$, measured in $G_n$, is properly defined. Let $\hR^{(n+1)}$  be the corresponding effective resistance 
  measured in $G_{n+1}$. \\

  \noindent If $u$ and $v$ belong to the same community, then $\hR^{(n+1)}$ satisfies
  \begin{equation}
    \hR^{(n)}_{u\, v}  - \hR^{(n+1)}_{u\, v} = \O{\frac{1}{\dbar^2}} \ge 0.
  \end{equation}
  If $u$ and $v$ are in different communities, then $\hR^{(n+1)}$ is controlled by the following inequalities,
  \begin{equation}
    \begin{cases}
      \hR^{(n)}_{u\, v}  - \hR^{(n+1)}_{u\, v} = \displaystyle\O{\frac{1}{\dbar^2}} \ge 0, 
      & \text{if} \quad k_n = k_{n+1},\\
      \\
      \hR^{(n)}_{u\, v}  - \hR^{(n+1)}_{u\, v} \ge \displaystyle \frac{2}{k_n^2} +
      \frac{1}{\dbar}\O{\frac{1}{k_n}} \ge 0 ,
      & \text{if} \quad k_{n+1} > k_n. 
    \end{cases}
  \end{equation}
  \label{lemma_ad_npo}
\end{lemma}
\begin{proof}
  Because of lemma \ref{n1_eq_n2}, we have 
  \begin{equation}
    \frac{1}{\dbarp} = \frac{1}{\dbar} + \O{\frac{1}{\dbarp^2}},
  \end{equation}
  and
  \begin{equation}
    \frac{1}{\dbarp^2} = \frac{1}{\dbar^2} + \O{\frac{1}{\dbarp^3}}.
  \end{equation}
  If $u$ and $v$ are in the same community, the expression for $\hR^{(n+1)}_{u\, v}$ and
  $\hR^{(n)}_{u\, v}$, given by \eqref{eq:theorem4} coincide, up to order $\O{1/\dbar^2}$.\\

  \noindent When $u$ and $v$ are in different communities, we need to consider the values of $k_n$ and
  $k_{n+1}$.  If $k_n = k_{n+1}$ then $\alpha (k_n,u,v) = \alpha (k_{n+1},u,v)$, and thus
  $\hR^{(n+1)}_{u\, v} = \hR^{(n)}_{u\, v}$, up to order $\O{1/\dbar^2}$. \\

  If $k_{n+1} > k_n$, we will show that the decrease in effective resistance is of order
  $\T{1/k_{n+1}^2}$. We first recall that $\hR^{(n)}_{u\, v} \ge \hR^{(n+1)}_{u\, v}$, since $G^{n+1}$ has
  more edges than $G_n$, and thus
  \begin{equation}
    \left|\hR^{(n)}_{u\, v} - \hR^{(n+1)}_{u\, v}\right| = \hR^{(n)}_{u\, v} - \hR^{(n+1)}_{u\, v}.
  \end{equation}
  Using the expression for $\hR^{(n+1)}_{u\, v}$ and $\hR^{(n)}_{u\, v}$, given by \eqref{eq:theorem4} we
  have
  \begin{equation}
    \begin{split}
      \left|\hR^{(n)}_{u\, v} - \hR^{(n+1)}_{u\, v}\right| 
      & = \hR^{(n)}_{u\, v} - \hR^{(n+1)}_{u\, v}\\
      & = \frac{1}{k_n} - \frac{1}{k_{n+1}} + \frac{\alpha(k_n,u,v)}{\dbar k_n} - \frac{\alpha(k_{n+1},u,v)}{\dbarp k_{n+1}} +
      \O{\frac{1}{\dbar^2}}\\ 
      & \ge \frac{1}{k_n} - \frac{1}{k_n+1} + \frac{\alpha(k_n,u,v)}{\dbar k_n} -
      \frac{\alpha(k_{n+1},u,v)}{\dbar (k_n+1)} + \O{\frac{1}{\dbar^2}}\\ 
      & \ge \frac{1}{k_n(k_n+1)} + \frac{1}{\dbar}\left(\frac{\alpha(k_n,u,v)}{k_n} - \frac{\alpha(k_{n+1},u,v)}{k_n+1} \right)+
      \O{\frac{1}{\dbar^2}}.\\
      & \ge \frac{1}{2k_n^2} + \frac{1}{\dbar}\left(\frac{\alpha(k_n,u,v)}{k_n} - \frac{\alpha(k_{n+1},u,v)}{k_n+1} \right)+
      \O{\frac{1}{\dbar^2}}.
    \end{split}
  \end{equation}
  We recall that corollary \ref{bounded-alpha-kn} implies that
  \begin{equation}
    \displaystyle 
    \left|\frac{\alpha(k_n+1,u,v)}{k_n+1} - \frac{\alpha(k_n,u,v)}{k_n} \right| = \O{\frac{1}{k_n^2}}
  \end{equation}
  and thus
  \begin{equation}
    \left|\hR^{(n)}_{u\, v} - \hR^{(n+1)}_{u\, v}\right| \ge \frac{1}{2k_n^2} + \frac{1}{\dbar}\O{\frac{1}{k_n}},
  \end{equation}
  which completes the proof.
\end{proof}
\subsection{The Distance  $D_n$ Under  the Null Hypothesis}
The following theorem provides an estimate of the distance $D_n$ between $G_n$ and $G_{n+1}$ 
after the addition of node $n+1$. Under the null hypothesis -- $n+1$ does not lead to an increase in the number of
cross-community edges -- the change in the normalized effective resistance distance between $G_n$
and $G_{n+1}$ remains negligible (after removing the linear term $H(n,k_n)$).
\begin{Theorem}
  \label{thm:upbound}
  Let $G_{n+1}\sim\cG(n+1,p_n,q_n)$ be a stochastic blockmodel with $p_n=\om{\log n/n}$ ,
  $q_n = \om{1/n^2}$, and $q_n = \o{p_n/n}$. Let $G_{n}$ be the subgraph induced by the vertex set
  $[n]$, and let $D_n = \RD\left(G_n,G_{n+1}\right)$ be the normalized effective resistance distance, $\RD$,  defined in
  \eqref{eq:metric-defn}.\\

  \noindent Suppose that the introduction of $n+1$ does not create additional cross-community edges,
  that is $k_n = k_{n+1}$, then 
  \begin{equation}
    \label{eq:116}
    D_n - h(n,k_n) =  \O{\frac{n^2}{\dbar^2}} \ge 0,
  \end{equation}
  where
  \begin{equation}
    h(n,k_n) = \left\lfloor \frac{n}{2}\right\rfloor + \left\lceil \frac{n}{2} \right\rceil
    \frac{k_n}{1+k_n}. 
    \label{h_n_kn}
  \end{equation}
\end{Theorem}
\begin{proof}
  Since vertex $n+1$ is isolated in $G_n$, the change in resistance at vertex $n+1$ between $G_n$
  and $G_{n+1}$ will behave quite differently than as described in Theorem \ref{thm:DRbound}. For
  this reason, we separate the renormalized resistance distance into two portions: the pairs of nodes
  that  do and do not contain vertex $n+1$,
  \begin{equation}
    D_n = \sum_{u < v \leq n+1}\left| R^{(n+1)}_{u\, v} - R^{(n)}_{u\, v}\right|
    = \sum_{u < v \leq n}\left| R^{(n+1)}_{u\, v} - R^{(n)}_{u\, v}\right| + \sum_{u \leq n}\left| R^{(n+1)}_{u\, n+1} - R^{(n)}_{u\, n+1}\right|.
    \label{eq:Dn_initial}
  \end{equation}
  Let us first study the second sum. Because vertex $n+1$ is isolated at time $n$,
  $R^{(n)}_{u\, n+1} = 1$. If $u$ and $n+1$ are in the same community, then
  $R^{(n+1)}_{u\, n+1} \leq \hR^{(n)}_{u\, v}= \O{1/\dbar^2}$, and is therefore negligible. We can use
  the trivial bound $R^{(n+1)}_{u\, n+1} \ge 0$ that yields
  \begin{equation*}
    \label{eq:18}
    \left| R^{(n+1)}_{u\, n+1} - R^{(n)}_{u\, n+1}\right| \leq 1.
  \end{equation*}
  If $u$ and $n+1$ are in different communities, then a tighter bound can be derived by considering
  the bottleneck formed by the cross-community edges. Indeed, a coarse application of Nash-Williams
  -- using only the cross-community cut-set -- tells us that the effective resistance between
  vertices in different communities is greater than $1/k_n$, and thus the renormalized effective
  resistance has the following lower bound,
  \begin{equation*}
    \label{eq:19}
    R^{(n+1)}_{u\, n+1} \ge \frac{1}{1+k_n},
  \end{equation*}
  which implies
  \begin{equation*}
    \left| R^{(n+1)}_{u\, n+1} - R^{(n)}_{u\, n+1}\right| \leq \frac{k_n}{1+k_n}.
  \end{equation*}
  Observing that there are $\left\lfloor \frac{n}{2}\right\rfloor$ possible in-community connections
  and $\left\lceil \frac{n}{2} \right\rceil$ possible cross-community connections, we have 
  \begin{equation}
    \label{eq:20}
    \sum_{u \leq n}\left| R^{(n+1)}_{u\, n+1} - R^{(n)}_{u\, n+1}\right| \leq \left\lfloor
      \frac{n}{2}\right\rfloor + \left\lceil \frac{n}{2} \right\rceil \frac{k_n}{1+k_n}.
  \end{equation}
  We now consider the first sum in \eqref{eq:Dn_initial}. Corollary \ref{lemma_ad_npo}
  combined with lemma \ref{lemma:change_bounds} yield
  \begin{equation}
    \left| R^{(n+1)}_{u\, v} - R^{(n)}_{u\, v}\right| 
    \leq \left| \hR^{(n+1)}_{u\, v} - \hR^{(n)}_{u\, v}\right| 
    = \O{\frac{1}{\dbar^2}},
  \end{equation}
  which leads to the following bound on the first sum,
  \begin{equation}
    \label{eq:35}
    \sum_{u<v\leq n}\left| R^{(n+1)}_{u\, v} - R^{(n)}_{u\, v}\right| = \O{\frac{n^2}{\dbar^2}}.
  \end{equation}
  Combining \eqref{eq:20} and \eqref{eq:35} yields 
  \begin{equation}
    0 \leq D_n = \mspace{-12mu} \sum_{u < v \leq n+1}\left| R^{(n+1)}_{u\, v} - R^{(n)}_{u\, v}\right| \leq \left\lfloor
      \frac{n}{2}\right\rfloor + \left\lceil \frac{n}{2} \right\rceil \frac{k_n}{1+k_n} +
    \O{\frac{n^2}{\dbar^2}} = h(n,k_n) + \O{\frac{n^2}{\dbar^2}},
  \end{equation}
  which implies the advertised result.
\end{proof}
The following corollary, which is  an immediate consequence of the previous theorem provides the
appropriate renormalization of $D_n - h(n,k_n)$  under the null hypothesis $H_0$.
\begin{corollary}
  \label{normalized_growth}
  Let $G_{n+1}\sim\cG(n+1,p_n,q_n)$ be a stochastic blockmodel with the same conditions on $p_n$ and
  $q_n$ as in Theorem \ref{thm:upbound}, and let $G_{n}$ be the subgraph induced by the vertex set
  $[n]$.\\

  \noindent Suppose that the introduction of $n+1$ does not create additional cross-community edges,
  that is $k_n = k_{n+1}$, then 
  \begin{equation}
    0 \leq p_n^2 \left( D_n - h(n,k_n) \right) = \O{1}
  \end{equation}
  \label{D_n_H0}
\end{corollary}
\begin{proof}
  As explained in lemma \ref{n1_eq_n2}, we assume without loss of generality that $n$ is even. We have then
  \begin{equation}
    \dbar^2 = \frac{(n/2-1)^2}{p_n^2}
  \end{equation}
  and thus
  \begin{equation}
    \frac{n^2}{\dbar^2} = \frac{n^2}{p^2_n(n/2 -1)^2} = \frac{4}{p_n^2}\left(1 + \O{\frac{1}{n}}\right),
  \end{equation}
  which leads to 
  \begin{equation}
    \frac{n^2p_n^2}{\dbar^2} = 4\left(1 + \O{\frac{1}{n}}\right).
  \end{equation}
  We recall that theorem \ref{thm:upbound} gives us the following bound on $\left( D_n - h(n,k_n)
  \right)$ under the null hypothesis,
  \begin{equation}
    \left( D_n - h(n,k_n) \right) = \O{\frac{n^2}{\dbar^2}},
  \end{equation}
  we conclude that 
  \begin{equation}
    p_n^2 \left( D_n - h(n,k_n) \right) = \O{1}.
  \end{equation}
\end{proof}
\subsection{The Distance $D_n$ Under the Alternate Hypothesis}
We now consider the case where the addition of node $n+1$ leads to an increase in the number of
cross-community edges. Loosening the bottleneck between the two communities creates a significant
change in the normalized effective resistance distance between $G_n$ and $G_{n+1}$.
\begin{Theorem}
  \label{thm:lowerbd}
  Let $G_{n+1}\sim\cG(n+1,p_n,q_n)$ be a stochastic blockmodel with $p_n=\om{\log n/n}$ ,
  $q_n = \om{1/n^2}$, and $q_n = \o{p_n/n}$. Let $G_{n}$ be the subgraph induced by the vertex set
  $[n]$, and let $D_n = \RD\left(G_n,G_{n+1}\right)$ be the normalized effective resistance
  distance, $\RD$, defined in \eqref{eq:metric-defn}.\\

  \noindent Suppose that the introduction of $n+1$ creates additional cross-community edges, that is
  $k_{n+1} > k_n$, then
  \begin{equation}
    \label{eq:6116}
    0 \leq \frac{1}{16}\left\{\frac{n^2}{k_n^2} + \frac{n^2}{\dbar}\O{\frac{1}{k_n}} \right\}
    \leq  D_n - h(n,k_n), 
  \end{equation}
  where $h(n,k_n)$ is defined in (\ref{h_n_kn}).
\end{Theorem}
\begin{proof} As before, we split the distance $D_n$ into two terms,
  \begin{equation}
    D_n  = \sum_{u < v \leq n+1}\left| R^{(n+1)}_{u\, v} - R^{(n)}_{u\, v}\right| 
    =  \sum_{u < v \leq n}\left| R^{(n+1)}_{u\, v} - R^{(n)}_{u\, v}\right| + 
    \sum_{1\leq i \leq n}\left| R^{(n+1)}_{u\, n+1} - R^{(n)}_{u\, n+1}\right|. 
    \label{eq:17}
  \end{equation}
  Again, we analyze the second sum, which will generate the same linear contribution,
  \begin{equation*}
    \label{eq:45}
    \left| R^{(n+1)}_{u\, n+1} - R^{(n)}_{u\, n+1}\right| = 1-R^{(n+1)}_{u\, n+1} = \frac{1}{1+ \hR^{(n+1)}_{u\, n+1}}. 
  \end{equation*}
  Because we seek an \emph{upper} bound on $R^{(n+1)}_{u\, n+1}$ to obtain an \emph{lower} bound on the
  change $\left| R^{(n+1)}_{u\, n+1} - R^{(n)}_{u\, n+1}\right|$ we have a more refined analysis of
  $R^{(n+1)}_{u\, n+1}$.\\

  \noindent In the case where $u$ and $n+1$ are in the same community, Theorem \ref{thm:cross-comm} tells us that
  \begin{equation*}
    \label{eq:47}
    \hR^{(n+1)}_{u\, n+1} = \frac{2}{\dbarp} + \O{\frac{1}{\dbar^2}},
  \end{equation*}
  We use the inequality $\frac{1}{1+x} \ge 1-x,$ which is valid for all $x>-1$, to get a
  lower bound,
  \begin{equation}
    \left| R^{(n+1)}_{u\, n+1} - R^{(n)}_{u\, n+1}\right| = \frac{1}{1+ \hR^{(n+1)}_{u\, n+1}}
    \ge 1 - \frac{2}{\dbarp} + \O{\frac{1}{\dbar^2}} 
    = 1 + \O{\frac{1}{\dbar}}.
    \label{eq:same_com_l}
  \end{equation}
  We also use the inequality $\frac{1}{1+x} \le 1-x/2$, which is valid for all $x \in [0,1]$, to get an
  upper bound,
  \begin{equation}
    \left| R^{(n+1)}_{u\, n+1} - R^{(n)}_{u\, n+1}\right| = \frac{1}{1+ \hR^{(n+1)}_{u\, n+1}}
    \le 1 - \frac{1}{\dbarp} + \O{\frac{1}{\dbar^2}} 
    = 1 + \O{\frac{1}{\dbar}}.
    \label{eq:same_com_u}
  \end{equation}
  \noindent   If $u$ and $n+1$ are in separate communities, Lemma \ref{lemma:cross-comm} tells us that
  \begin{equation}
    \frac{1}{k_{n+1}} \leq    \hR^{(n+1)}_{u\, n+1}\leq \frac{1}{k_{n+1}} + \frac{4}{\dbarp} + \O{\frac{1}{\dbarp^2}} \leq
    \frac{1}{1+k_n} + \frac{4}{\dbarp} + \O{\frac{1}{\dbar^2}}.
  \end{equation}
  Using again the inequality $\frac{1}{1+x} \ge 1-x,$ we get a  lower bound,
  \begin{equation}
    \left| R^{(n+1)}_{u\, n+1} - R^{(n)}_{u\, n+1}\right| 
    \ge 1 - \frac{1}{1+k_n} - \frac{4}{\dbarp} + \O{\frac{1}{\dbar^2}}
    = \frac{k_n}{1+k_n} + \O{\frac{1}{\dbar}},
    \label{eq:diff_com_l}
  \end{equation}
  and  using the inequality $\frac{1}{1+x} \le 1-x/2$ we get an  upper bound,
  \begin{equation}
    \left| R^{(n+1)}_{u\, n+1} - R^{(n)}_{u\, n+1}\right| \leq 1 - \frac{1}{2k_{n+1}} \leq 1.
    \label{eq:diff_com_u}
  \end{equation}
  Combining \eqref{eq:same_com_l}, \eqref{eq:same_com_u},\eqref{eq:diff_com_l}, and \eqref{eq:diff_com_u}, we get
  \begin{equation}
    \left\lfloor
      \frac{n}{2}\right\rfloor + \left\lceil \frac{n}{2} \right\rceil \frac{k_n}{1+k_n} +
    \O{\frac{n}{\dbar}} 
    \leq \sum_{u \leq n}\left| R^{(n+1)}_{u\, n+1} - R^{(n)}_{u\, n+1}\right| \le 
    n +   \O{\frac{n}{\dbar}}. 
    \label{eq:sum1}
  \end{equation}
  We now consider the first sum in \eqref{eq:17}. To get lower and upper bounds on
  $\left| \hR^{(n+1)}_{u\, v} - \hR^{(n)}_{u\, v}\right|$ we use lemma \ref{lemma:change_bounds}.

  We first observe that for $n$ sufficiently large, we have $\hR^{(n)}_{u\, v} \leq 1$, and thus
  $\hR^{(n+1)}_{u\, v} \leq 1$. Combining this upper bound on the effective resistance with lemma
  \ref{lemma:change_bounds} we get
  \begin{equation}
    \text{if} \quad  C(u,v) \le \left| \hR^{(n+1)}_{u\, v} - \hR^{(n)}_{u\, v}\right| 
    \quad \text{then} \quad 
    \frac{C(u,v)}{4} \le \frac{C(u,v)}{(1 + \hR^{(n)}_{u\, v})(1 + \hR^{(n+1)}_{u\, v})} 
    \le \left| \hR^{(n+1)}_{u\, v} - \hR^{(n)}_{u\, v}\right|. 
  \end{equation}
  From corollary \ref{lemma_ad_npo} we have
  \begin{equation}
    \begin{cases}
      C(u,v) = \O{\frac{1}{\dbar^2}} & \text{if $u$ and $v$ are in the same community,}\\
      C(u,v) \ge \displaystyle \frac{1}{k_n^2} + \frac{1}{\dbar}\O{\frac{1}{k_n}}  
      & \text{otherwise,}
    \end{cases}
  \end{equation}
  and therefore
  \begin{equation}
    4 \mspace {-12mu} \sum_{u < v \leq n}\left| R^{(n+1)}_{u\, v} - R^{(n)}_{u\, v}\right| \ge 
    \mspace{-12mu} \sum_{\stackrel{\scriptstyle u,v \in \text{different}}{\scriptstyle
        \text{communities}} } \mspace{-24mu} C(u,v) \mspace{8mu} + \mspace{-4mu}
    \sum_{\stackrel{\scriptstyle u,v \in \text{same}}{\scriptstyle \text{community}} }
    \mspace{-12mu} C(u,v) 
    \ge 
    \frac{n^2}{4k_n^2} + \frac{n^2}{4\dbar}\O{\frac{1}{k_n}} + \frac{n^2}{4} \O{\frac{1}{\dbar^2}},
  \end{equation}
  where the differences between $n/2$ and the exact size of $C_1$ or $C_2$ are absorbed in
  the error terms. Also, we have
  \begin{equation}
    \frac{1}{\dbar^2} = \frac{1}{\dbar}\O{\frac{1}{k_n}}, 
  \end{equation}
  and thus
  \begin{equation}
    \sum_{u < v \leq n}\left| R^{(n+1)}_{u\, v} - R^{(n)}_{u\, v}\right| \ge 
    \frac{1}{16}\left\{\frac{n^2}{k_n^2} + \frac{n^2}{\dbar}\O{\frac{1}{k_n}}\right\}.
    \label{eq:sum2}
  \end{equation}
  Finally, we note that 
  \begin{equation}
    \frac{n}{\dbar} = \frac{n^2}{\dbar k_n}\frac{k_n}{n} = \frac{n^2}{\dbar k_n}\frac{k_n}{n} 
  \end{equation}
  Because of lemma \ref{k_n_concentrates}, we have asymptotically with high probability,
  \begin{equation}
    \frac{n}{\dbar} =  \frac{n^2}{\dbar k_n}\O{\frac{\E{k_n}}{n}} = \frac{n^2}{\dbar k_n} \O{nq} = \frac{n^2}{\dbar} \O{\frac{1}{k_n}},
  \end{equation}
  and thus we conclude that 
  \begin{equation}
    \O{\frac{n}{\dbar}} =  \frac{n^2}{\dbar} \O{\frac{1}{k_n}}.
  \end{equation}
  Lastly, we add the two sums \eqref{eq:sum1} and \eqref{eq:sum2} to get
  \begin{equation}
    \left\lfloor
      \frac{n}{2}\right\rfloor + \left\lceil \frac{n}{2} \right\rceil \frac{k_n}{1+k_n} +
    \frac{1}{16}\left\{\frac{n^2}{k_n^2} + \frac{n^2}{\dbar}\O{\frac{1}{k_n}} \right\}
    \le \sum_{u < v \leq n+1}\left| R^{(n+1)}_{u\, v} - R^{(n)}_{u\, v}\right|. 
    \label{almost-there}
  \end{equation}
  The leading term linear term, $h(n,k_n)$, in (\ref{almost-there}) can be subtracted to arrive at
  the advertised result.
\end{proof}
\noindent Using the same normalization described in corollary \ref{normalized_growth} we obtain a very
different growth for $p_n^2(D_n - h(n,k_n))$ in the case of the alternate hypothesis.
\begin{corollary}
  Let $G_{n+1}\sim\cG(n+1,p_n,q_n)$ be a stochastic blockmodel with the same conditions on $p_n$ and
  $q_n$ as in Theorem \ref{thm:upbound}, and let $G_{n}$ be the subgraph induced by the vertex set
  $[n]$.\\

  \noindent Suppose that the introduction of $n+1$ creates additional cross-community edges, that is
  $k_{n+1} > k_n$, then 
  \begin{equation}
    0 \leq p_n \left( D_n - h(n,k_n) \right) \rightarrow \infty \quad \text{with high probability}.
  \end{equation}
  \label{D_n_H1}
\end{corollary}
\begin{proof}
  As explained in lemma \ref{n1_eq_n2}, we assume without loss of generality that $n$ is even. From
  \eqref{eq:6116} we have 
  \begin{equation}
    p_n\left(D_n - h(n,k_n)\right) \ge \frac{1}{16}\left\{\frac{(np_n)^2}{k_n^2} +
      \frac{(np_n)^2}{\dbar}\O{\frac{1}{k_n}} \right\}
    \label{pnDn}
  \end{equation}
  Without loss of generality we assume $n$ even, and we have
  \begin{equation}
    \frac{(np_n)^2}{\dbar}\O{\frac{1}{k_n}} = \frac{(np_n)^2}{p_n(n/2-1)q_n (n^2/4)}\O{1}
    = \frac{p_n}{nq_n}\O{1} = \om{1}\O{1}.
  \end{equation}
  Therefore the second term in (\ref{pnDn}) is either bounded, or goes to infinity. We will prove
  that the first term goes to infinity.  We have
  \begin{equation}
    \frac{np_n}{k_n} = \frac{np_n}{\E{k_n}} \frac{\E{k_n}}{k_n} 
    = \frac{np_n}{q_n(n/2)^2} \frac{\E{k_n}}{k_n}  = \frac{4p_n}{nq_n}\frac{\E{k_n}}{k_n}. 
  \end{equation}
  From lemma \ref{k_n_concentrates} we know that asymptotically $\E{k_n}/k_n =
  \T{1}$ with high probability. Also, we have $p_n/(nq_n) = \om{1}$. This concludes the proof.
\end{proof}
The quantity $p_n^2(D_n - h(n,k_n))$ could provide a statistic to test the null hypothesis $k_n =
k_{n+1}$ against the alternate hypothesis $k_n < k_{n+1}$. Unfortunately, computing $p_n^2(D_n -
h(n,k_n))$ requires the knowledge of the unknown parameter $p_n$, and unknown variable $h(n,k_n)$. We
therefore propose two estimates that converge to these unknowns. A simple estimate of  $h(n,k_n)$ is
provided by $n$. Since we assume that there are much fewer cross-community edges than edges within
each community, we can estimate $p_n$ from the total number of edges. 

We start with two technical lemmas. The first lemma shows that can replace $h(n,k_n)$ with $n$.
\begin{lemma}
  \label{h_equal_n}
  Let $G_n\sim\cG(n,p_n,q_n)$ be a stochastic blockmodel with the same conditions on $p_n$ and
  $q_n$ as in Theorem \ref{thm:upbound}. If $p_n = \O{1/\sqrt{n}}$, then we have
  \begin{equation}
    \lim_{n\rightarrow \infty} p_n^2 \left(n -h(n,k_n)\right) = 0 \quad \text{with high probability}.
  \end{equation}
\end{lemma}
\begin{proof}
  We have 
  \begin{equation}
    n - h(n,k_n) = \left\lceil \frac{n}{2} \right\rceil \frac{1}{k_n+1}.
  \end{equation}
  Because $\lceil n/2\rceil = (n/2) \T{1}$, we have
  \begin{equation}
    p_n^2(n - h(n,k_n)) = \T{1} \frac{np_n^2}{2(k_n+1)} = \frac{np_n^2}{2\E{k_n}}
    \frac{\E{k_n}}{k_n}\frac{k_n}{k_n+1} \T{1}.
  \end{equation}
  Now, we have $k_n/(k_n + 1) < 1$,  $n p_n^2 = \O{1}$, and $\E{k_n} = \om{1}$, therefore 
  \begin{equation}
    \lim_{n\rightarrow \infty} \frac{np_n^2}{2\E{k_n}}\frac{k_n}{k_n+1} \T{1} = 0.
  \end{equation}
  Finally, we recall that $\E{k_n}/k_n = \T{1}$ with high probability, which concludes the proof.
\end{proof}
We now consider the estimation of $p_n$.
\begin{lemma}
  \label{m_n_equal_p_n}
  Let $G_n\sim\cG(n,p_n,q_n)$ be a stochastic blockmodel with the same conditions on $p_n$ and
  $q_n$ as in Theorem \ref{thm:upbound}. Let $m_n$ be the total number of edges in $G_n$. Then the
  probability $p_n$ can be estimated asymptotically from $m_n$ and $n$,
  \begin{equation}
    \frac{4m_n}{n^2} = p_n \left(1 + \O{1/n}\right),
    \quad \text{with high probability}.
  \end{equation}
\end{lemma}
\begin{proof}
  The proof proceeds in two steps. We first show that $k_n$ concentrates around its expectation
  $\E{k_n}$,  and then we argue that $\lim_{n\rightarrow \infty} 4\E{m_n}/n^2= p_n$.\\

  \noindent The total number of edges, $m_n$, in the graph $G_n$, can be decomposed as
  \begin{equation}
    m_n = m_{n_1} + m_{n_2} + k_n,
  \end{equation}
  where $m_{n_1}$ ($m_{n_2}$) is the number of edges in community $C_1$ ($C_2$). The three random
  variables are binomial (with different parameters), and they concentrate around their respective
  expectations. Consequently $m_n$ also concentrates around its expectation, and we can combine
  three Chernoff inequalities using a union bound to show that
  \begin{equation}
    \frac{m_n}{\E{m_n}} = \T{1}, \quad \text{with high probability.}
    \label{m_n_concentrates}
  \end{equation}

  \noindent A quick computation of $\E{m_n}$ shows that
  \begin{equation}
    \E{m_n} = p_n\frac{n^2}{4}\left(1 - \frac{2}{n} + \frac{q_n}{p_n} + \o{\frac{1}{n^2}}\right).
  \end{equation}
  Also, $q_n/p_n = \o{1/n}$, and thus
  \begin{equation}
    \E{m_n} = p_n\frac{n^2}{4}\left(1 + \O{\frac{1}{n}}\right).
    \label{Em_n}
  \end{equation}
  To conclude, we combine (\ref{m_n_concentrates}) and (\ref{Em_n}), to get
  \begin{equation}
    \frac{4m_n}{n^2} = \frac{m_n}{\E{m_n}} \frac{4\E{m_n}}{n^2} = p_n \left(1 - \O{\frac{1}{n}}\right),
  \end{equation}
  which concludes the proof.
\end{proof}
We define the following statistic that asymptotically converges toward $p_n^2 (D_n -
h(n.k_n))$ with high probability, as explained in the next theorem.
\begin{definition}
  Let $G_{n+1}\sim \cG(n,p_n,q_n)$ be a stochastic blockmodel with the same conditions on $p_n$ and
  $q_n$ as in Theorem \ref{thm:upbound}. Let $G_{n}$ be the subgraph induced by the vertex set
  $[n]$. Let $D_n = \RD\left(G_n,G_{n+1}\right)$ be the normalized effective resistance distance,
  $\RD$, defined in \eqref{eq:metric-defn}.\\

  \noindent We define the statistic
  \begin{equation}
    Z_n \eqdef \frac{16m_n^2}{n^4} \left(D_n - n \right).
    \label{the_statistic}
  \end{equation}
\end{definition}
\begin{Theorem}
  Let $G_n\sim\cG(n,p_n,q_n)$ be a stochastic blockmodel with the same conditions on $p_n$ and
  $q_n$ as in Theorem \ref{thm:upbound}. If $p_n = \O{1/\sqrt{n}}$, then we have
  \begin{equation}
    Z_n = p_n^2 \left(D_n - h(n,k_n)\right) \left (1 + \o{1}\right), \quad \text{ with high probability.}
  \end{equation}
  \label{D_n_boostrap}
\end{Theorem}
\begin{proof}
  The proof is an elementary consequence of the two lemmas \ref{h_equal_n} and \ref{m_n_equal_p_n}. We
  have
  \begin{equation}
    \frac{16m_n^2}{n^2} \left(D_n - n \right) 
    =  \frac{16m_n^2}{n^4} \left(D_n - h(n,k_n) \right) + \frac{16m_n^2}{n^4} \left(h(n,k_n) - n
    \right).
  \end{equation}
  Using lemma \ref{m_n_equal_p_n}, we have
  \begin{equation}
    \frac{16m_n^2}{n^2} \left(D_n - n \right) 
    =  p_n^2\left(1 + \O{1/n} \right)^2 \left(D_n - h(n,k_n)     \right)  
    + \left(1 + \O{1/n} \right)^2 p_n^2 \left(h(n,k_n) - n       \right).
  \end{equation}
  Lemma \ref{h_equal_n} shows that the second term can be neglected,
  \begin{equation}
    \begin{split}    
      \frac{16m_n^2}{n^2} \left(D_n - n \right) 
      &=  p_n^2\left(1 + \O{1/n}\right) \left(D_n - h(n,k_n) \right)
      + \left(1 + \O{1/n}  \right) \o{1}\\
      &=  p_n^2\left(D_n - h(n,k_n) \right)\left(1 + \O{1/n} \right)  + \o{1}.
    \end{split}
  \end{equation}
  Because $p_n^2\left(D_n - h(n,k_n) \right)$ is either bounded, or goes to infinity, we have
  \begin{equation}
    \frac{(16m_n^2/n^2)(D_n - n)}{p_n^2\left(D_n - h(n,k_n) \right)}
    = 1 + \o{1}, 
  \end{equation}
  which concludes the proof.
\end{proof}
We finally arrive at the main theorem.
\begin{Theorem}
  Let $G_{n+1}\sim \cG(n,p_n,q_n)$ be a stochastic blockmodel with the same conditions on $p_n$ and
  $q_n$ as in Theorem \ref{thm:upbound}. Let $G_{n}$ be the subgraph induced
  by the vertex set $[n]$.\\

  \noindent To test the hypothesis
  \begin{equation}
    H_0: \quad k_n = k_{n+1}
  \end{equation}
  versus 
  \begin{equation}
    H_1: \quad k_n < k_{n+1}
  \end{equation}
  we use the test based on the statistic $Z_n$ defined in (\ref{the_statistic}) where we accept
  $H_0$ if $Z_n < z_\varepsilon$ and accept $H_1$ otherwise. The threshold $z_\varepsilon$ for the rejection
  region satisfies
  \begin{equation}
    \proba_{H_0}\left(Z_n \ge z_\varepsilon\right) \leq \varepsilon \quad \text{as}\quad n \rightarrow \infty,
  \end{equation}
  and
  \begin{equation}
    \proba_{H_1}\left(Z_n \ge z_\varepsilon\right) \rightarrow 1 \quad \text{as}\quad n \rightarrow \infty.
  \end{equation}
  The test has therefore asymptotic level $\varepsilon$ and asymptotic power 1.
\end{Theorem}
\begin{proof}
  Assume $H_0$ to be true. Because of corollary  \ref{D_n_H0} and Theorem \ref{D_n_boostrap},
  \begin{equation}
    Z_n = \O{1} , \quad \text{with high probability.}
  \end{equation}
  In other words, for every $0 < \varepsilon < 1$ there exists $z_\varepsilon$ such that
  \begin{equation}
    \prob{Z_n <z_\varepsilon} = 1 - \varepsilon, \quad \text{as} \quad n \rightarrow \infty,
  \end{equation}
  or
  \begin{equation}
    \prob{Z_n \ge z_\varepsilon} = \varepsilon, \quad \text{as} \quad n \rightarrow \infty.
  \end{equation}
  Assume now $H_1$ to be true. Because of corollary  \ref{D_n_H1} and Theorem \ref{D_n_boostrap},
  \begin{equation}
    Z_n = \om{1} , \quad \text{with high probability}
  \end{equation}
  Therefore, for every $0 < \gamma < 1$, there exists $n_0$ such that 
  \begin{equation}
    \forall n \ge n_0, \prob{Z_n > z_\varepsilon} = 1 - \gamma, \quad \text{as} \quad n \rightarrow \infty.
  \end{equation}
  In other words,
  \begin{equation}
    \proba_{H_1}\left(Z_n \ge z_\varepsilon\right) \rightarrow 1 \quad \text{as}\quad n \rightarrow \infty,
  \end{equation}
  which concludes the proof.
\end{proof}
\printbibliography
\end{document}